\documentclass[a4paper,11pt]{article}

\usepackage[printonlyused]{acronym}
\usepackage{epsfig,amssymb,amsfonts,amsmath,amsthm}
\usepackage[ngerman,main=american]{babel}
\usepackage[format=plain,position=bottom,labelfont=bf,textfont=normal]{caption}
\usepackage{dsfont}
\usepackage[inline]{enumitem}
\usepackage{fancyhdr}
\usepackage{float}
\usepackage{graphicx}
\usepackage[mathlines,switch*]{lineno}
\usepackage[final,tracking=true,kerning=true,spacing=true,factor=1100,stretch=10,shrink=10]{microtype}
\usepackage[left=0.9in,right=0.9in,top=1in,bottom=1in]{geometry}
\usepackage[labelformat=simple]{subcaption}
\usepackage{titlesec}
\usepackage{todonotes}
\usepackage{versions}
\usepackage{wrapfig}
\usepackage{xspace}

\usepackage{paralist}
\setlist[enumerate]{itemsep=0mm}

\usepackage{algorithm}

\usepackage{fancybox}
\usepackage{xcolor}

\usepackage[ocgcolorlinks]{hyperref} 

\colorlet{DarkRed}{red!75!black}
\colorlet{DarkGreen}{green!75!black}
\colorlet{DarkBlue}{blue!75!black}

\hypersetup{
	linkcolor = DarkRed,
	citecolor = DarkGreen,
	urlcolor = DarkBlue,
	bookmarksnumbered = true,
	linktocpage = true
}

\newcommand{\mylemma}[2]{\begin{lem}\label{lem:#1}#2\end{lem}}

\newcommand{\mytheorem}[2]{\begin{thm}\label{thm:#1}#2\end{thm}}

\newcommand{\againtheorem}[2]{\noindent\textbf{Theorem~\ref{thm:#1}}
	(from page \pageref{thm:#1})\textbf{.}\emph{#2}}

\def\dash---{\kern.16667em---\penalty\exhyphenpenalty\hskip.16667em\relax}

\newcommand{\poly}{\operatorname{poly}}

\newcommand{\R}{\mathds{R}}

\newcommand{\B}{B}
\newcommand{\BT}{B^{\rot}}

\newcommand{\mF}{\mathcal{F}}

\newcommand{\F}{F}
\newcommand{\FT}{F^{\rot}}
\newcommand{\U}{U}
\newcommand{\UT}{U^{\rot}}
\newcommand{\I}{I}

\renewcommand{\leq}{\leqslant}
\renewcommand{\geq}{\geqslant}
\renewcommand{\le}{\leqslant}
\renewcommand{\ge}{\geqslant}

\newcommand{\thmref}[1]{Theorem~\ref{thm:#1}}

\newcommand{\lemref}[1]{Lemma~\ref{lem:#1}}

\newcommand{\norm}[1]{\left\Vert #1 \right\Vert_{2}}

\renewcommand{\epsilon}{\varepsilon}
\newcommand{\eps}{\varepsilon}

\newcommand{\A}{\mathcal{A}}

\newcommand{\Def}{\overset{\mathrm{def}}{=}}

\floatname{algorithm}{Problem}

\newcommand{\rot}{\mathrm{T}}
\newcommand{\abs}[1]{| #1 |}

\newcommand{\SPSD}{\mathrm{SPSD}}

\newcommand{\Tr}{\mathrm{Tr}}
\newcommand{\lk}{\lambda_{k}}
\newcommand{\lkp}{\lambda_{k+1}}

\newcommand{\cX}{\mathcal{X}}

\newcommand{\tX}[1]{\triangle_{{#1}}(\cX)}

\newcommand{\tY}[1]{\triangle_{{#1}}(Y)}
\newcommand{\tXV}[1]{\triangle_{{#1}}(\Yp)}

\newcommand{\twXVs}[1]{\triangle_{{#1}}(\wYp)}

\newcommand{\wX}{\widetilde{X}}

\newcommand{\Xp}{X^{\prm}}
\newcommand{\XpT}{(X^{\prm})^{\rot}}
\newcommand{\wXp}{\widetilde{X^{\prm}}}
\newcommand{\wXpT}{(\widetilde{X^{\prm}})^{\rot}}

\newcommand{\wXa}{\widetilde{X_{\alpha}}}
\newcommand{\wXaP}{\widetilde{X_{\alpha}^{\prm}}}

\newcommand{\XpOpt}{X_{\mathrm{opt}}^{\prm}}
\newcommand{\wXpOpt}{\widetilde{X_{\mathrm{opt}}^{\prm}}}

\newcommand{\XOptPrm}[1]{X_{\mathrm{opt}}^{\prm({#1})}}

\newcommand{\wSig}{\widetilde{\Sigma}}
\newcommand{\wV}{\widetilde{V}}
\newcommand{\wU}{\widetilde{U}}

\newcommand{\Unk}{U_{n-k}}

\newcommand{\Yp}{Y^{\prm}}
\newcommand{\wY}{\widetilde{Y}}
\newcommand{\wYp}{\widetilde{Y^{\prm}}}

\newcommand{\YpT}{(Y^{\prm})^{\rot}}

\newcommand{\wfi}{\widehat{f_{i}}}
\newcommand{\wfj}{\widehat{f_{j}}}

\newcommand{\dprm}{\delta^{\prm}}
\newcommand{\dPsi}{\delta}

\newcommand{\hrhok}{\widehat{\rho}(k)}
\newcommand{\hrhokp}{\widehat{\rho}(k+1)}
\newcommand{\hrAvrK}{\widehat{\rho}_{\mathrm{avr}}(k)}

\newcommand{\LG}{\mathcal{L}_G}
\newcommand{\Ups}{\Upsilon}
\newcommand{\APR}{\alpha}
\newcommand{\prm}{\prime}

\newcommand{\ORSS}{\mathrm{ORSS}}

\renewcommand{\vec}[1]{\boldsymbol{\mathbf{#1}}}

\makeatletter

\protected\def\internalcacs{\let\AC@hyperlink\@secondoftwo\acs}

\protected\def\internalcacsp{\let\AC@hyperlink\@secondoftwo\acsp}

\makeatother

\newtheorem{thm}{Theorem}  

\newtheorem{lem}[thm]{Lemma}

\newtheorem{cor}[thm]{Corollary}

\floatname{algorithm}{Algorithm}
\floatname{figure}{Figure}
\floatname{table}{Table}

\nolinenumbers

\SetTracking{encoding=*, shape=sc}{3} 

\title{Approximate Spectral Clustering:\\
		Efficiency and Guarantees\footnote{
		A preliminary version of this paper was presented at the 
		24th Annual European Symposium on Algorithms (ESA 2016).}}

\author{
	Pavel Kolev\footnote{This work has been funded by the Cluster of Excellence ``Multimodal Computing and Interaction" within the Excellence Initiative of the German Federal Government.} \qquad Kurt Mehlhorn\\
	Max Planck Institute for Informatics, Saarland Informatics Campus, Germany\\
	\{pkolev,mehlhorn\}@mpi-inf.mpg.de
}

\date{}

\begin{document}

	\thispagestyle{empty}
	\maketitle

	\begin{abstract}
		Approximate Spectral Clustering (ASC) is a popular and
		successful heuristic for partitioning the nodes of a graph $G$
		into clusters for which the ratio of outside connections 
		compared to the volume (sum of degrees) is small.
		ASC consists of the following two subroutines:
		i) compute an approximate Spectral Embedding via the Power method; and
		ii) partition the resulting vector set with an approximate
		$k$-means clustering algorithm. The resulting $k$-means partition naturally 
		induces a $k$-way node partition of $G$.
		
		We give a comprehensive analysis of ASC building on the work of
		Peng et al.~(SICOMP'17), Boutsidis et al.~(ICML'15) and Ostrovsky et al.~(JACM'13).
		We show that ASC 
		i) runs efficiently, and 
		ii) yields a good approximation of an optimal $k$-way node partition of $G$.
		Moreover, we strengthen the quality guarantees of a structural result of Peng et al.
		by a factor of $k$, and simultaneously weaken the eigenvalue gap assumption. 
		Further, we show that ASC finds a $k$-way node partition of $G$ with the 
		strengthened quality guarantees.
	\end{abstract}

	\newpage
	\tableofcontents
	\setcounter{page}{0}
	\newpage

	\section{Introduction}\label{sec:SSCPIntro}

A \emph{cluster} in an undirected graph $G = (V,E)$ is a subset $S$ of nodes
whose volume is large compared to the number of outside connections.
Formally, the \emph{conductance} of $S$ is defined as
\[
	\phi(S) \Def \frac{|E(S,\overline{S})|}{\min\{\mu(S),\mu(\overline{S})\}},
\]
where the volume of $S$ is given by $\mu(S)\Def\sum_{v\in S}\deg(v)$. 
We are interested in the problem of partitioning the nodes into 
a given number $k$ of clusters in a way that (approximately) minimizes the 
\emph{$k$-way partition constant}
\begin{equation}\label{eq:defhrho}
	\hrhok  \Def \min_{\text{partition $(P_1,\dots,P_k)$ of $V$}}\,\,\max_{i\in\{1,\dots,k\}}\,\,\phi(P_i).
\end{equation}
The 2-way partitioning constant is also known as the conductance of the graph 
and is denoted as 
\begin{equation}\label{eq:GraphConductanceProblem}
\phi_G \Def \min_{S\subseteq V} \phi(S).
\end{equation}
The $k$-way partitioning problem arises in many applications, e.g., 
image segmentation and exploratory data analysis.
We refer to the survey~\cite{Luxburg07} for additional information. 
Further, the surveys~\cite{ShiM00,KannanVV04,Luxburg07} discuss properties 
of graphs with small or large conductance.

\subsubsection*{Hardness and Approximation}

The $k$-way partitioning problem is known to be $\mathsf{NP}$-hard,
even for $k=2$ \cite{MatulaS90}.
In the case when $k=2$, the $k$-way partitioning problem reduces to
the graph conductance problem~\eqref{eq:GraphConductanceProblem},
for which there is an approximation algorithm~\cite{ChungBook}
that computes a bipartition $(S,\overline{S})$ such that $\phi(S)\leq\sqrt{2\phi_G}$.
The algorithm computes an eigenvector corresponding to the second smallest eigenvalue
of a normalized Laplacian matrix, sorts the eigenvector's entries, and performs
a sweep over the sorted vector. The best set is guaranteed to satisfy the 
approximation bound.

This indicates that the second eigenvector encodes sufficient information
for computing an approximate bipartition and motivated researchers to consider
the bottom $k$ eigenvectors in order to approximately solve the 
$k$-way partitioning problem. The resulting approach is called Spectral Clustering.

\subsubsection*{Spectral Clustering}

Given an undirected graph $G=(V,E)$ and a number of clusters $k$,
the Spectral Clustering algorithm consists of the following two steps:
\begin{compactenum}[\mbox{}\hspace{\parindent}(i)]
	\item Compute the bottom $k$ eigenvectors of the normalized Laplacian matrix of $G$ and store them into a matrix $Y\in\R^{n\times k}$ and interpret the $i$-th row as a vector in $\R^k$ corresponding to the $i$-th node of the input graph. This step is known as Spectral Embedding (SE).
	
	\item partition the resulting vector set into $k$ clusters  using a $k$-means clustering algorithm.
\end{compactenum}\smallskip

Numerous works report of the practical success of Spectral Clustering 
in solving challenging clustering problems, and in particular it has been successfully 
applied in the fields of image segmentation, pattern recognition, 
data mining, community detection and VLSI design
\cite{AlpertY95, ShiM00, NgJW01, MalikBLS01, BelkinN01, LiuZ04, Zelnik-ManorP04, 
	WhiteS05, Luxburg07, WangD12, Tasdemir12, CucuringuKCMP16}.

\subsubsection*{Approximate Spectral Clustering}

Exact computation of Spectral Clustering is expensive due to the following two bottlenecks:
\begin{compactenum}[\mbox{}\hspace{\parindent}(i)]
	\item the best algorithm for computing a SE exactly requires time 
	$\Omega(n^{\omega})$, cf.~\cite{Woodruff2014};
	
	\item the $k$-means clustering problem is NP-hard~\cite{MahajanNV12}.
\end{compactenum}
\smallskip\smallskip

It is therefore necessary to relax the preceding two problems and 
to focus on designing approximation schemes for them.
Several approximation techniques were developed for Spectral Clustering
\cite{Prenter1981,SpielmanT14,YanHJ09,CaoCDL14,FowlkesBCM04,PavanP04, 
	BezdekHHLR06,WangLRB09,nystrom1930,WangD12,Tasdemir12,LinC10,Woodruff2014}. 

The Power method~\cite{LinC10,Woodruff2014} is perhaps the most popular technique
for computing an approximate SE, due to its simplicity and 
ease of implementation.
Further, this technique was successfully applied for 
low-rank matrix approximation~\cite{Woodruff2014},
and it has a worst case convergence guarantee in terms of a principal angle between 
the space spanned by the approximate and the true eigenvectors~\cite[Theorem 8.2.4]{Golub1996}.

Although, the $k$-means clustering problem is NP-hard~\cite{MahajanNV12}, it admits 
a polynomial time approximation scheme (PTAS)~\cite{KSS04, Har-PeledK05, FeldmanMS07, ORSS12}. 
However, the best PTAS for computing a $(1+\eps)$ approximation incurs a factor 
$2^{\poly(k/\eps)}$ in the runtime.

On the other hand, it is folklore that the approximate variant of Spectral Clustering
which computes an approximate SE via the Power method,
and applies to it an approximate $k$-means clustering algorithm, 
recovers a good approximation of an optimal $k$-way node partition of $G$ 
and at the same time runs efficiently (in nearly-linear time).

It is an important task for theory to explain the practical success of
Approximate Spectral Clustering, 
and in particular to resolve the following three questions.
In order to state them, we need some notation.
Let $Y$ be a SE computed exactly, and $\wY$ be an 
approximate SE computed via the Power method. 
Further, let $X$ ($\wX$) be an optimal $k$-means clustering partition of 
the rows of $Y$ ($\wY$).
Let $\wXa$ be a $k$-way row partition of $\wY$, 
computed by an $\alpha$-approximate $k$-means clustering algorithm.
The following questions arise:

\smallskip\smallskip\smallskip
\begin{compactenum}[\mbox{}\hspace{\parindent}(1)]
	\item[Q1.] Show that $\wXa$ is a good approximation of $X$.
	\smallskip\smallskip
	
	\item[Q2.] Show that the $k$-way node partition of $G$ induced by $\wXa$,
	yields a good approximation of an optimal $k$-way node partition of $G$.
	\smallskip\smallskip
	
	\item[Q3.] Show that Approximate Spectral Clustering runs efficiently
	(in nearly-linear time).
\end{compactenum}

\subsubsection*{Eigenvalue Gaps and $k$-Way Partitions}

Let $0=\lambda_1\leq \ldots\leq\lambda_n\leq 2$ be the eigenvalues of 
a normalized Laplacian matrix of $G$. It was observed experimentally~\cite{Luxburg07,F2010} 
that a large gap between $\lambda_{k+1}$ and $\lambda_{k}$ guarantees 
a good $k$-way node partition of $G$ and this was formally proven 
in~\cite{LeeGT12,GharanT14}. Lee, Gharan and Trevisan~\cite{LeeGT12} studied the 
$k$-way expansion constant defined as 
\begin{equation}\label{eq:defrho}
	\rho(k)\Def\min_{\text{disjoint }S_1,\dots,S_k}\,\,\max_{i\in\{1,\dots,k\}}\,\, \phi(S_i),	
\end{equation}
and related it to $\lambda_k$ via higher-order Cheeger inequalities
\begin{equation}\label{eq:highorder}
	\lambda_k/2\leq\rho(k)\leq O(k^2)\sqrt{\lambda_k}.
\end{equation}
For related works on higher-order Cheeger inequalities, we refer the reader to~\cite{LouisRTV12,KwokLL17}. 
Gharan and Trevisan~\cite{GharanT14} showed that the $k$-way partition constant is at most a factor $k$ away from the $k$-way expansion constant, i.e., 
\begin{equation}\label{eq:hrhoLEQkrho}
	\rho(k)\leq \hrhok  \leq k\cdot\rho(k).
\end{equation}

In particular, \eqref{eq:highorder} and \eqref{eq:hrhoLEQkrho} together yield that 
$\lambda_{k+1}\gg O(k^3)\sqrt{\lambda_{k}}$ implies $\hrhokp\gg\hrhok $. 
Thus, there is a $k$-way node partition $(P_1,\dots,P_k)$ of $G$ such that 
$\phi(P_i)\leq O(k^3)\sqrt{\lambda_{k}}$ for all $i$,
and simultaneously the best $(k+1)$-way partition is significantly worse.

\subsubsection*{Prior Work}

Ng et al.~\cite{NgJW01} reported that ASC performs very well on
challenging clustering instances, and initiated the study
for finding a formal explanation for the practical success of ASC.
Using tools from matrix perturbation theory,
they derived sufficient conditions under which
the vectors of a SE form tight clusters. 
However, their analysis does not apply to approximate SEs,
and does not give guarantees for the induced $k$-way node partition of $G$.

Peng et al.~\cite{Peng0Z17} showed that for all instances satisfying
the eigenvalue gap assumption $\lambda_{k+1}/\hrhok \geq\Omega(k^3)$,
any $O(1)$-approximate $k$-means partition of a normalized SE $\Yp$
induces a good approximation of an optimal $k$-way node partition of $G$.
Notably, their analysis yields the first approximation guarantees in terms of 
the $k$-way partition constant.
However, their analysis does not apply to approximate SE,
and also computing an $O(1)$-approximation $k$-means partition
using any known PTAS~\cite{Har-PeledK05, FeldmanMS07,ORSS12} 
incurs an exponential factor of $2^{\Omega(k)}$ in the running time.

Boutsidis et al.~\cite{BoutsidisKG15} showed that an approximate $k$-means partition of
an approximate SE $\wY$ computed via the Power method,
yields a $k$-means partition $P$ of the exact SE $Y$
such that the $k$-means cost of $P$ yields an additive approximation 
to the optimum $k$-means cost of $Y$.
This gives an affirmative answer to question Q1. Further, the authors stated 
as main open problems to resolve questions Q2 and Q3.

Besides designing a PTAS for the $k$-means clustering problems,
Ostrovsky et al.~\cite{ORSS12} analyzed a variant of 
Lloyd $k$-means clustering algorithm.
They showed that on input a set of $n$ vectors in $\R^k$ satisfying 
a natural well-clusterable assumption, the algorithm \emph{efficiently} computes
a good approximation of an optimal $k$-means partition. 
In particular, the algorithm runs in time $O(k^2(n+k^2))$.

A natural question to ask is whether the analysis of Peng et al.~\cite{Peng0Z17},
Boutsidis et al.~\cite{BoutsidisKG15} and Ostrovsky et al.~\cite{ORSS12} 
can be integrated and extended to answer the questions Q2 and Q3?

\subsubsection*{Our Contribution: An Overview}

We give a comprehensive analysis of ASC building on the work of
Peng et al.~\cite{Peng0Z17}, Boutsidis et al.~\cite{BoutsidisKG15} 
and Ostrovsky et al.~\cite{ORSS12}.
We show that the Approximate Spectral Clustering
i) runs efficiently, and 
ii) yields a good approximation of an optimal $k$-way node partition of $G$.
Moreover, we strengthen the quality guarantees of a structural result of 
Peng et al.~\cite{Peng0Z17} by a factor of $k$, and simultaneously 
weaken the eigenvalue gap assumption.
Further, our analysis shows that the Approximate Spectral Clustering
finds a $k$-way node partition of $G$ with the strengthened quality guarantees.
This gives an affirmative answer to questions Q2 and Q3.

	\subsection{Notation}

\subsubsection*{$k$-means Clustering Problem}
Let $\cX$ be a set of vectors of the same dimension. Then
\[
	\tX{k} \Def \min_{\text{partition }(X_{1},\dots,X_{k}) \text{ of } \cX}\, 
	\sum_{i=1}^{k}\sum_{x \in X_{i}} \left\Vert x -c_{i}\right\Vert_{2}^{2}, 
	\quad\text{where}\quad 
	c_i \Def \frac{1}{\abs{X_i}} \sum_{x \in X_i} x,
\]
is the optimum cost of clustering $\cX$ into $k$ sets.
A $k$-means partition $(X_1,\dots,X_k)$ of $\cX$, 
with corresponding gravity centers  $c_1,\dots,c_k$ as above,
is $\APR$-approximate if
\begin{equation}
	\mathrm{Cost}(\{ X_{i}\} _{i=1}^{k}) \Def  
	\sum_{i=1}^{k}\sum_{x \in X_{i}} \left\Vert x -c_{i}\right\Vert_{2}^{2}
	\leq \APR \cdot \triangle_k(\cX).
\end{equation}
Given a matrix $Y$, we abuse notation and write $\tY{k}$ to denote
the optimum $k$-means cost of partitioning the rows of $Y$.

\subsubsection*{Spectral Embeddings}\label{subsec:Notation}

Given an undirected graph $G=(V,E)$ with $m = \abs{E}$ edges and
$n = \abs{V}$ nodes, let $D$ be the diagonal degree matrix and $A$ be the adjacency matrix.
Then, the graph Laplacian matrix is defined as $L=D-A$, and the 
normalized Laplacian matrix is given by $\LG=I-\A$, where $\A=D^{-1/2}AD^{-1/2}$.
Further, let $f_i \in \R^V$ be the eigenvector corresponding to the $i$-th 
smallest eigenvalue $\lambda_i$ of $\LG$.

The \emph{canonical} Spectral Embedding, for short \emph{canonical} SE, 
is defined as a matrix $Y\in\R^{n\times k}$ composed of 
the bottom $k$ eigenvectors~\footnote{
The Eigendecomposition theorem guarantees that all eigenvectors are orthonormal.} 
of $\LG$ corresponding to the $k$ smallest eigenvalues. 
The \emph{approximate} SE is computed via the Power method~\footnote{
	Given a symmetric matrix $M$ and a number $k$, 
	the Power method approximates the top $k$ eigenvectors of $M$ corresponding to 
	the largest $k$ eigenvalues.
	Since we seek a good approximation of the bottom $k$ eigenvectors of $\LG=I-\A$,
	associated with the smallest $k$ eigenvalues, we initialize the Power method 
	with $M=I+\A$.}.
Namely,
let $S\in\R^{n\times k}$ be a matrix whose entries are i.i.d. samples from the 
standard Gaussian distribution $N(0,1)$ and $p$ be the number of iterations. 
Then, the \emph{approximate} SE $\wY$ is given by:
\begin{equation}\label{eq:defWY}
	1)\,\, M \Def I+\A	;\quad 2)\,\,\text{Let }\wU\wSig\wV^{\rot}
	\text{ be the SVD}~\footnote{
		SVD abbreviates Singular Value Decomposition, see~\cite{Woodruff2014}.}
	\text{ of }M^{p}S; \quad\text{and}\quad 3)\,\, \wY \Def \wU \in \R^{n\times k}.
\end{equation}

Peng et al.~\cite{Peng0Z17} do not apply $k$-means directly to the 
\emph{canonical} SE, but first normalize it by dividing 
the row corresponding to $u$ by $\sqrt{d(u)}$ and then 
put $d(u)$ copies of the resulting vector into the $k$-means clustering instance. 
This repetition of vectors is crucial for their analysis, in order to 
achieve approximation guarantees in terms of volume overlap and conductance. 
We follow their approach. 

We construct a matrix $\Yp \in\R^{2m\times k}$
such that for every node $u\in V$, we insert $d(u)$ many copies of 
the normalized row $Y(u,:)/\sqrt{d(u)}$ to $\Yp $. 
Formally, the \emph{normalized} SE $\Yp$ and the
\emph{approximate normalized} SE $\wYp$ are defined by
\begin{equation}\label{eq:defYpWYp}
	\Yp \Def \left(\begin{array}{c}
	\vec{1}_{d(1)}\frac{Y(1,:)}{\sqrt{d(1)}}\\
	\cdots\\
	\vec{1}_{d(n)}\frac{Y(n,:)}{\sqrt{d(n)}}
	\end{array}\right)_{2m\times k} \quad\text{and}\quad 
	\wYp \Def \left(\begin{array}{c}
	\vec{1}_{d(1)}\frac{\wY (1,:)}{\sqrt{d(1)}}\\
	\cdots\\
	\vec{1}_{d(n)}\frac{\wY (n,:)}{\sqrt{d(n)}}
	\end{array}\right)_{2m\times k},
\end{equation}
where $\vec{1}_{d(i)}$ is the all-one column vector with dimension $d(i)$.

We can assume w.l.o.g. that a $k$-means clustering algorithm applied on 
$\Yp$ ($\wYp$), outputs a $k$-means partition such that all copies of row 
$Y(v,:)/\sqrt{d(v)}$ ($\wY(v,:)/\sqrt{d(v)}$) belong to the same cluster, 
for all nodes $v$. Thus, the algorithm induces a $k$-way node partition of $G$.

\subsection{Our Contribution}

A key prior structural result, established by Peng et al.~\cite{Peng0Z17}, 
connects the normalized SE $\Yp$, 
$\alpha$-approximate $k$-means clustering, 
the $k$-way partition constant $\hrhok$, see~\eqref{eq:defhrho}, 
and the $(k+1)$-st eigenvalue $\lkp$ of the normalized Laplacian matrix $\LG$.
In particular, they proved the following statement under 
a gap assumption defined in terms of
\[
	\Ups\Def\frac{\lkp}{\hrhok}.
\]

\begin{thm}\cite[Theorem $1.2$]{Peng0Z17}\label{thm:Peng0Z17}
	Let $k\geq3$ and $G$ be a graph satisfying the gap assumption~\footnote{
		Note that $\lk /2 \le \hrhok$, see \eqref{eq:defYpWYp}. 
		Thus, the assumption implies $\lk /2 \le  \hrhok \le \delta \lkp/(2 \cdot 10^5 \cdot k^3)$, 
		i.e., there is a substantial gap between the $(k+1)$-th and the $k$-th eigenvalue.}
	\begin{equation}\label{eq:Ups}
		\delta \Def 2\cdot10^{5}\cdot k^3/\Ups \leq 1/2.
	\end{equation}
	Let $(P_1,\dots,P_k)$ be a $k$-way node partition of $G$ achieving $\hrhok$, and
	let $(A_1,\dots,A_k)$ be the $k$-way node partition of $G$ induced by 
	an $\APR$-approximate $k$-means partition of the normalized SE $\Yp$.
	Then, for every $i\in\{1,\dots,k\}$ it hold (after suitable renumbering of one of the partitions) that
	\[
		1)\,\, \mu(A_i\triangle P_i) \le \APR\delta\cdot\mu(P_i)
		\quad\text{and}\quad 2)\,\, \phi(A_i) \leq (1+2\APR\delta)\cdot\phi(P_i) + 2\APR\delta.
	\]
\end{thm}

Under a stronger eigenvalue gap assumption 
$2\cdot10^{5}\cdot k^{5}/\Upsilon\leq1/2$, 
Peng et al.~\cite{Peng0Z17} gave an algorithm that finds in time 
$O\left(m\cdot\poly \log(n)\right)$ a $k$-way node partition of $G$ 
with essentially the guarantees stated in Theorem~\ref{thm:Peng0Z17}.
However, their algorithmic result substitutes normalized SE 
with Heat Kernel Embedding and $k$-means clustering with locality sensitive hashing.
Thus, the algorithmic part of their paper does not explain the success of 
Approximate Spectral Clustering.

We give affirmative answer to the questions Q2 and Q3.
On the way, we also strengthen the approximation guarantees in Theorem~\ref{thm:Peng0Z17}
by a factor of $k$ and simultaneously weaken the eigenvalue gap assumption.

Let ${\cal{O}}$ be the set of all $k$-way partitions $(P_1,\dots,P_k)$ achieving
the $k$-way partition constant $\hrhok$. Let
\[
	\hrAvrK \Def \min_{ (P_1,\dots,P_k) \in {\cal O}} \frac{1}{k} \sum_{i=1}^{k} \phi(P_i)
\]
be the \emph{minimum average conductance} over all $k$-way partitions in ${\cal O}$.
Note that $\hrAvrK\leq\hrhok$.
Our gap assumption is defined in terms of
\[
	\Psi \Def \frac{\lkp}{\hrAvrK}.
\]
For the remainder, we denote by $(P_1,\dots,P_k)$ a $k$-way node partition of $G$ 
achieving $\hrAvrK$. 

We present now our main result, consisting of a structural and an algorithmic statement.

\begin{thm}\label{thm:myPSZ15}
	a) (Existence of a Good Clustering) 
	Let $k\geq3$ and $G$ be a graph satisfying
	\begin{equation}\label{eq:Psi}
		\dPsi \Def 20^{4}\cdot k^3 / \Psi \leq 1/2.
	\end{equation}
	Let $(P_1,\dots,P_k)$ be a $k$-way node partition of $G$ achieving $\hrAvrK$, and
	let $(A_1,\dots,A_k)$ be the $k$-way node partition of $G$ induced by 
	an $\APR$-approximate $k$-means partition of the normalized SE $\Yp$.
	Then, for every $i\in\{1,\dots,k\}$ it hold (after suitable renumbering of one of the partitions) that
	\[
		1)\,\,\mu(A_{i}\triangle P_{i})\leq\frac{\APR\dPsi}{10^{3}k}\cdot\mu(P_{i})
		\quad\text{and}\quad
		2)\,\,\phi(A_{i})\leq\left(1+\frac{2\APR\dPsi}{10^{3}k}\right)\cdot\phi(P_{i})+
		\frac{2\APR\dPsi}{10^{3}k}.
	\]
	b) (An Efficient Algorithm) If in addition $k/\dPsi\geq10^{9}$ and~\footnote{
		The case $\tXV{k}\le n^{-O(1)}$ constitutes a trivial clustering problem. 
		For technical reasons, we have to exclude too easy inputs.}
		$\tXV{k}\geq n^{-O(1)}$, then the variant of Lloyd algorithm analyzed by 
		Ostrovsky et al.~\cite{ORSS12} when applied to the approximate normalized SE
		$\wYp$, induces in time $O(m(k^{2}+\frac{\ln n}{\lkp}))$ with constant probability
		a $k$-way node partition $(A_{1},\dots,A_{k})$ of $G$ 
		such that for every $i\in\{1,\dots,k\}$
		it hold (after suitable renumbering of one of the partitions) that
	\[
	3)\,\,\mu(A_{i}\triangle P_{i})\leq\frac{2\dPsi}{10^{3}k}\cdot\mu(P_{i})\quad\text{and}\quad 4)\,\,\phi(A_{i}) \leq\left(1+\frac{4\dPsi}{10^{3}k}\right)\cdot\phi(P_{i})+\frac{4\dPsi}{10^{3}k}.
	\]
\end{thm}

Part (a) of Theorem~\ref{thm:myPSZ15} strengthens the quality guarantees 
in Theorem~\ref{thm:Peng0Z17} by a factor of $k$, and simultaneously
weaken the eigenvalue gap assumption.
Part (b) of Theorem~\ref{thm:myPSZ15} gives a comprehensive analysis
of Approximate Spectral Clustering,
and demonstrates that the algorithm 
i) runs efficiently, and 
ii) yields a good approximation of an optimal $k$-way node partition of $G$.\\
Further, it shows that the Approximate Spectral Clustering 
finds a $k$-way node partition of $G$ with the strengthened quality guarantees,
and whenever $k\leq (\log n)^{O(1)}$ and $\lambda_{k+1} \ge 1/(\log n)^{O(1)}$,
the algorithm runs in nearly linear time.
This answers affirmatively questions Q2 and Q3.

\subsubsection*{Remarks}
The variant of Lloyd $k$-means clustering algorithm, analyzed by Ostrovsky et al.~\cite{ORSS12}, 
is efficient only for inputs $\cX$ satisfying $\tX{k} \leq \eps^{2}\tX{k-1}$
for some $\eps\in(0,\eps_0]$, where $\eps_0 = 6/10^{7}$.
The authors stated that their result should also hold for a larger $\eps_0$,
and mentioned that they did not attempt to maximize $\eps_0$.

An anonymous reviewer of the conference version of this paper,
suggested to include a numerical example.
Consider a graph which consists of $k$ cliques each of size $n/k$, 
plus $k$ additional edges that connect the cliques in the form of a ring. 
This graph is a trivial clustering instance, and for any constant $k$ 
it holds that~\footnote{
	A graph $G$ has $k$ connected components iff $\lambda_k=0$. 
	For any clique $K_n$, we have $\lambda_1=0$ and $\lambda_2=\dots=\lambda_n=1$.
	Further, when $G$ consists of $k$ cliques $K_{n/k}$ disconnected from each other,
	then $\lambda_{k}=0$ and $\lambda_{k+1}=\dots=\lambda_n=1$.}  
$\lambda_{k+1}\rightarrow1$ and $\lambda_{k}\rightarrow0$.
Observe that $\hrAvrK = \hrhok \approx (k/n)^2$. 
For the gap assumption to hold we need $\lambda_{k+1} \ge 2 \cdot 20^4 \cdot k^3 \cdot \hrAvrK$. 
This implies $n\geq\sqrt{2\cdot20^4\cdot k^5/\lambda_{k+1}}$.
For small $k$, this is a modest requirement on the size of the graph.

For the algorithmic result, we need in addition $\delta \le k \cdot \eps_0 /600$. 
For the gap condition to hold, we need 
$\lambda_{k+1} \ge (600/\eps_0 k) \cdot 20^4 \cdot k^3 \cdot (k^2/n^2)$ or
$n\ge\sqrt{600\cdot20^{4}\cdot k^{4}/(\eps_{0}\lambda_{k+1})}$. 
For $\eps_0 = 6 / 10^{7}$, this amounts to 
$n\ge\sqrt{2{}^{4}\cdot10^{13}\cdot k^{4}/\lambda_{k+1}}$, 
a quite large lower bound on $n$.

The statement that Part (b) of Theorem~\ref{thm:myPSZ15} gives a theoretical support 
for the practical success of Approximate Spectral Clustering,
therefore has to be taken with a grain of salt.
It is only an asymptotic statement and does not explain the good behavior on small graphs.

\subsection{Our Techniques}

In Section~\ref{sec:ImprovedStructuralResult}, 
we give a refined spectral analysis of \cite{Peng0Z17} which yields
the improved structural result in Part (a) of Theorem~\ref{thm:myPSZ15}.
In Section~\ref{sec:AnalysisOfApproximateSpectralClustering},
we connect Part (a) of Theorem~\ref{thm:myPSZ15} with the work of 
Ostrovsky et al.~\cite{ORSS12} and Boutsidis et al.~\cite{BoutsidisKG15},
yielding the algorithmic result in Part (b) of Theorem~\ref{thm:myPSZ15}.

Ostrovsky et al.~\cite{ORSS12} analyzed a variant of Lloyd $k$-means clustering algorithm.
We refer to this algorithm as the $\ORSS$ clustering algorithm.
The $\ORSS$-algorithm is efficient only for inputs $\cX$ satisfying:
some partition into $k$ clusters is much better than any partition 
into $k-1$ clusters. Formally, it states

\begin{thm}\cite[Theorem 4.15]{ORSS12}\label{thm_KMs}
	Assuming that $\tX{k}\leq\eps^{2}\cdot\tX{k-1}$
	for $\eps\in(0,6\cdot10^{-7}]$, the $\ORSS$-algorithm 
	runs in time $O(nkd+k^{3}d)$ and returns
	with probability at least $1-O(\sqrt{\eps})$
	a $k$-way partition of $\cX$ with cost at most
	$[(1-\eps^{2})/(1-37\eps^{2})]\tX{k}$.
\end{thm}

Let $Z\in\R^{n\times k}$ be a matrix and $(R_1,\dots,R_k)$ be a row partition of $Z$.
Let $c_{j}=\tfrac{1}{|R_{j}|}\sum_{u\in R_{j}}Z_{u,:}$ be the gravity center of cluster $R_j$, 
for all $j\in\{1,\dots,k\}$.
We next express in matrix notation the $k$-means cost of partition $(R_1,\dots,R_k)$.
To this end, we introduce an \emph{indicator} matrix $X\in\mathbb{R}^{n\times k}$ 
such that $X_{ij}=1/\sqrt{|R_j|}$ if row $Z_{i,:}$ belongs to cluster $R_j$, 
and $X_{ij}=0$ otherwise.
Then, $(XX^TZ)_{i,:}=c_{j}$, where row $Z_{i,:}$ belongs to cluster $R_j$.
Hence, the $k$-means cost of $(R_1,\dots,R_k)$ becomes
\begin{equation}\label{eq:ZXXTZeqCost}
\mathrm{Cost}(\{ R_{i}\}_{i=1}^{k})
= \sum_{j=1}^{k}\sum_{u\in R_{j}}\lVert Z_{u,:} c_{j}\rVert _{2}^{2}
= \lVert Z-XX^{\rot}Z\rVert _{F}^{2}.
\end{equation}

\subsubsection*{Our Analytical Approach}

Our main technical contribution is to prove that the 
approximate normalized SE $\wYp$ computed via the Power method
is $\eps$-separated, i.e. the assumption $\twXVs{k}<\eps^{2}\cdot\twXVs{k-1}$
of Ostrovsky et al.~\cite{ORSS12} is satisfied.
This implies, by Theorem~\ref{thm_KMs}, that 
the $\ORSS$-algorithm runs efficiently on $\wYp$.
Let the resulting $k$-way row partition of $\wYp$ be encoded by
the indicator matrix $\wXp$.

Then, building on the work of \cite{BM14,BoutsidisKG15},
we show that $\wXp$ is a good approximation of an optimal 
$k$-means partition of the corresponding normalized SE $\Yp$.
Further, using our strengthened structural result in 
Part (a) of Theorem~\ref{thm:myPSZ15},
we show that $\wXp$ induces a good approximation of an optimal
$k$-way node partition of graph $G$, 
in terms of volume overlap and conductance.

First, we establish in Section~\ref{sec:OrgSEisEpsSep} the assumption of 
Ostrovsky et al.~\cite{ORSS12} for the normalized SE $\Yp$.

\newcommand{\thmGapTriK}
{
	(normalized SE is $\eps$-separated)
	Let $G$ be a graph that satisfies $\Psi=20^{4}\cdot k^{3}/\dPsi$,
	$\dPsi\in(0,1/2]$ and $k/\dPsi\geq10^{9}$. 
	Then for $\eps=6\cdot10^{-7}$ it holds
	$\tXV{k}\leq\eps^{2}\cdot \tXV{k-1}$.
}
\mytheorem{thmGapTriK}{\thmGapTriK}

\thmref{thmGapTriK} does not suffice for proving Part (b) of Theorem~\ref{thm:myPSZ15}, 
since it requires the analogous statement for the approximate normalized SE $\wYp$.

In Subsection~\ref{subsec:thmMyPWM}, we show that an $\alpha$-approximate 
$k$-means clustering algorithm applied to the approximate normalized SE $\wYp$, 
yields an approximate $k$-way row partition of 
the corresponding normalized SE $\Yp$.

\newcommand{\thmMyPWM}
{
	(Similar to \cite[Theorem 6]{BoutsidisKG15}, but analyzes the
	approximate normalized SE)
	Let $\eps,\delta_p\in(0,1)$ be arbitrary.
	Compute the approximate normalized SE $\wYp$ via the Power method
	with $p\geq\ln(8nk/\epsilon\delta_p)\big/\ln(1/\gamma_{k})$ iterations
	and $\gamma_{k}=(2-\lambda_{k+1})/(2 -\lk)<1$.
	Run on the rows of $\wYp$ an $\alpha$-approximate $k$-means clustering algorithm 
	with failure probability $\delta_{\alpha}$.
	Let the outcome be a clustering indicator matrix $\wXaP\in\R^{n\times k}$.
	Then, with probability at least $1 - 2e^{-2n} - 3\delta_p-\delta_{\alpha}$,
	it holds that
	\[
		\lVert \Yp -\wXaP(\wXaP)^{\rot}\Yp \rVert _{F}^{2} \leq
		(1+4\eps)\cdot \alpha \cdot \tXV{k} + 4\eps^{2}.
	\]
}
\mytheorem{thmMyPWM}{\thmMyPWM}

In Subsection~\ref{subsec:ThmPartTwo}, using \thmref{thmGapTriK} and \thmref{thmMyPWM},
we show that the approximate normalized SE $\wYp$ satisfies 
the assumption of Ostrovsky et al.~\cite{ORSS12}.

\newcommand{\thmEasySpectralEmbedding}
{
	(approximate normalized SE is $\eps$-separated)
	Assume $\Psi=20^{4}\cdot k^{3}/\dPsi$,
	$k/\dPsi\geq10^{9}$ for some $\dPsi\in(0,1/2]$ and
	the optimum $k$-means cost of the normalized SE $\Yp$ is such that~\footnote{
		$\tXV{k}\geq n^{-O(1)}$ asserts a multiplicative 
		approximation guarantee in \thmref{thmMyPWM}.}
	$\tXV{k}\geq n^{-O(1)}$. 
	Compute the approximate normalized SE $\wYp$ via the Power method
	with $p\geq\Omega(\frac{\ln n}{\lkp})$.
	Then, for $\eps=6\cdot10^{-7}$ it holds with high probability that 
	$\twXVs{k}<5\eps^{2}\cdot\twXVs{k-1}$.
}
\mytheorem{thmEasySpectralEmbedding}{\thmEasySpectralEmbedding}

Finally, in Subsection~\ref{subsec:ThmPartThree}, we prove Part (b) of Theorem~\ref{thm:myPSZ15}
by combining Part (a) of Theorem~\ref{thm:myPSZ15}, Theorem~\ref{thm_KMs}, \thmref{thmMyPWM}
and \thmref{thmEasySpectralEmbedding}.

	\newpage
	
	\section{Improved Structural Result}\label{sec:ImprovedStructuralResult}

\subsection{Notation}\label{subsec:ESEN}

We use the notation adopted in~\cite{Peng0Z17}. 
Let $\lambda_{j}$ be the $j$-th eigenvalue 
of the normalized Laplacian matrix $\LG$, and let $f_j\in\R^{V}$ be the 
associated eigenvector ($\LG f_j=\lambda_{j} f_j$).

Let $\overline{g_{i}}=\frac{D^{1/2}\chi_{P_{i}}}{\left\Vert D^{1/2}\chi_{P_{i}}\right\Vert_2}$, 
where $\chi_{P_{i}}$ is the characteristic vector of the subset $P_i\subseteq V$. Note that $\overline{g_{i}}$ is the normalized characteristic vector of $P_i$ and $\left\Vert D^{1/2}\chi_{P_{i}}\right\Vert_2^2 = \sum_{v\in P_{i}}d(v) = \mu(P_i)$. The Rayleigh quotient is defined by and satisfies
\[
\mathcal{R}\left(\overline{g_{i}}\right)\Def\frac{\overline{g_{i}}^{\rot}\LG\overline{g_{i}}}{\overline{g_{i}}^{\rot}\overline{g_{i}}} = \frac{1}{\mu(P_{i})}\chi_{P_{i}}^{\rot}L\chi_{P_{i}} = \frac{|E(S,\overline{S})|}{\mu(P_{i})}=\phi(P_{i}),
\]
where the Laplacian matrix $L=D-A$ and the normalized Laplacian matrix
$\LG=D^{-1/2}LD^{-1/2}$.

The eigenvectors $\{f_i\}_{i=1}^{n}$ form an orthonormal basis of $\R^n$. Thus each characteristic vector $\overline{g_{i}}$ can be expressed as $\overline{g_{i}}=\sum_{j=1}^{n}\alpha_{j}^{(i)}f_{j}$ for all $i\in\{1,\dots,k\}$. We define its \emph{projection} onto the first $k$ eigenvectors by $\wfi=\sum_{j=1}^{k}\alpha_{j}^{(i)}f_{j}$.

Peng et al.~\cite{Peng0Z17} proved that if the gap parameter $\Upsilon$ 
is large enough then $\mathrm{span}(\{\wfi\}_{i=1}^{k})=\mathrm{span}(\{f_{i}\}_{i=1}^{k})$ 
and the first $k$ eigenvectors can be expressed by $f_{i}=\sum_{j=1}^{k}\beta_{j}^{(i)}\wfj$, 
for all $i\in\{1,\dots,k\}$. Moreover, they demonstrated that each vector $\widehat{g_{i}}=\sum_{j=1}^{k}\beta_{j}^{(i)}\overline{g_{j}}$ approximates 
the eigenvector $f_i$, for all $i\in\{1,\dots,k\}$. 
We will show that similar statements 
hold with weakened gap parameter $\Psi$.

The \emph{estimation centers} induced by the canonical SE are given by
\begin{equation}\label{def of pi} p^{(i)}=\frac{1}{\sqrt{\mu(P_{i})}}\left(\beta_{i}^{(1)},\dots,\beta_{i}^{(k)}\right)^\rot.
\end{equation}

Our analysis crucially relies on spectral properties of the following two matrices. 
Let $\F,\B\in\R^{k\times k}$ be square matrices defined by
\begin{equation}\label{eq:mtxFandB}
\F_{j,i}=\alpha_{j}^{(i)}\quad\text{and}\quad \B_{j,i}=\beta_{j}^{(i)}.
\end{equation}
In Figure~\ref{fig:relation}, we show the relation among the vectors 
$f_i$, $\wfi$, $\widehat{g_i}$ and $\overline{g_i}$.

\begin{figure}[h]
	\centering
	\begin{tikzpicture}
	
	\node (a) at (0, 0) {{$f_i=\sum_{j=1}^{k}\beta_{j}^{(i)}\wfj$}};
	
	\draw [-, color=darkgray,thick] (0,0.5) -- (0,2.5);
	
	\node (b) at (0, 3) {{$\wfi=\sum_{j=1}^{k}\alpha_{j}^{(i)}f_j$}};
	
	\node (c) at (8.15, 0) {{$\widehat{g_i} = \sum_{j=1}^{k}\beta_{j}^{(i)}\overline{g_j}$}};
	
	\node (d) at (9, 3) {{$\overline{g_i} = \frac{D^{1/2}\chi_{P_i}}{\sqrt{\mu(P_i)}} = \sum_{j=1}^{n}\alpha_{j}^{(i)}f_j$}};
	
	\draw [-, color=darkgray,thick] (8,2.5) -- (8,0.5);
	\draw [ color=darkgray,thick] (1.7,0) -- (6.5,0);
	\draw [ color=darkgray,thick] (1.7,3) -- (6.5,3);
	
	\node (g) at (4, 3.25) {\footnotesize{$\|\wfi-\overline{g_i}\|_2^2 \leq \phi(P_i)/\lkp$}};
	
	\node (h) at (4,2.65) {};
	
	\node (i) at (4, 0.25) {\footnotesize{$\|f_i - \widehat{g_i}\|_2^2 \le (1+3k/\Psi)\cdot k/\Psi$}};
	
	\node (k) at (4,-0.35) {};
	\end{tikzpicture}
	\caption{The vectors $\{f_i\}_{i = 1}^n$ are eigenvectors of the normalized Laplacian matrix $\LG$.
		The vectors $\{\overline{g}_i\}_{i = 1}^k$ are the normalized characteristic vectors of 
		an optimal partition $(P_1,\dots,P_k)$. 
		For each $i\in\{1,\dots,k\}$ the vector $\wfi$ is the projection of vector
		$\overline{g_i}$ onto $\mathrm{span}(f_1,\ldots,f_k)$. 
		The vectors $\wfi$ and $\overline{g}_i$ are close for $i \in\{1,\dots,k\}$. 
		If $\Psi>4\cdot k^{3/2}$, then
		$\mathrm{span}(f_1,\ldots,f_k) = \mathrm{span}(\widehat{f_1},\ldots,\widehat{f_k})$
		and thus we can write $f_i = \sum_{j = 1}^k \beta^{(i)}_j \wfj$. 
		Further, the vectors $f_i$ and $\widehat{g_i} = \sum_{j = 1}^k \beta_j^{(i)} \overline{g_j}$ 
		are close for $i\in\{1,\dots,k\}$.}
	\label{fig:relation}
\end{figure}

\subsection{Technical Insights}\label{subsec:ESEN}

The analysis of Part (a) of Theorem~\ref{thm:myPSZ15} follows the proof approach 
in \cite[Theorem 1.2]{Peng0Z17}, but improves upon it in essential ways.

Our first technical insight is that the matrices $\BT \B$ and 
$\B\BT $ are close to the identity matrix. We prove this in two steps. 
In Section~\ref{overlineg and f are close}, we show that the vectors $\widehat{g_i}$ and 
${f_i}$ are close, and then in Section~\ref{sec:AMB} we analyze the column 
space and row space of matrix $\B$.

\begin{thm}[Matrix $\B\BT $ is Close to Identity Matrix]\label{thmBFC}
	If $\Psi\geq 10^4\cdot k^{3}/\eps^{2}$ and $\eps\in(0,1)$ then for all distinct 
	$i,j\in\{1,\dots,k\}$ it holds
	\[
	1-\eps\leq\left\langle \B_{i,:},\B_{i,:}\right\rangle \leq1+\eps\quad\text{and}\quad\left|\left\langle \B_{i,:},\B_{j,:}\right\rangle \right|\leq\sqrt{\eps}.
	\]
\end{thm}

Using Theorem~\ref{thmBFC}, we give a strengthened version of \cite[Lemma $4.2$]{Peng0Z17} 
that depends on the weaken gap parameter $\Psi$.

\begin{lem}\label{lem_pi2}
	If $\Psi=20^4\cdot k^3/\dPsi$ for some $\dPsi\in(0,1]$ then for every $i \in \{1,\dots,k\}$ it holds that
	\[
		\left(1-\sqrt{\dPsi}/4\right)\frac{1}{\mu(P_{i})}
		\leq \left\Vert p^{(i)}\right\Vert_{2}^{2} \leq
		\left(1+\sqrt{\dPsi}/4\right)\frac{1}{\mu(P_{i})}.
	\]
\end{lem}

\begin{proof}
	By definition $p^{(i)}=\frac{1}{\sqrt{\mu(P_{i})}}\cdot\B_{i,:}$ and Theorem~\ref{thmBFC} yields $\left\Vert \B_{i,:}\right\Vert_{2}^{2} \in [1\pm\sqrt{\dPsi}/4]$.
\end{proof}

Using Theorem~\ref{thmBFC} and Lemma~\ref{lem_pi2}, 
we establish a strengthened version of \cite[Lemma $4.3$]{Peng0Z17} 
that depends on the weaken gap parameter $\Psi$, and simultaneously
shows that the $\ell_2$ distance between estimation centers is
larger by a factor of $k$.

\begin{lem}[Larger Distance Between Estimation Centers]\label{lem:pi are well-spread}
	If $\Psi= 20^{4}\cdot k^3/\dPsi$ for some $\dPsi\in(0,\frac{1}{2}]$ then for any distinct $i,j\in\{1,\dots,k\}$ it holds that
	\[
		\left\Vert p^{(i)}-p^{(j)}\right\Vert_{2}^{2}\geq
		\left[2\cdot\min\left\{ \mu(P_{i}),\mu(P_{j})\right\} \right]^{-1}.
	\]
\end{lem}

\begin{proof}
	Since $p^{(i)}$ is a row of matrix $B$, Theorem~\ref{thmBFC} with $\eps=\sqrt{\dPsi}/4$ yields
	\[
	\left\langle \frac{p^{(i)}}{\left\Vert p^{(i)}\right\Vert_{2}},\frac{p^{(j)}}{\left\Vert p^{(j)}\right\Vert_{2}}\right\rangle
	= \frac{\left\langle \B_{i,:},\B_{j,:}\right\rangle }{\left\Vert \B_{i,:}\right\Vert_{2}\left\Vert \B_{j,:}\right\Vert_{2}} \le \frac{\sqrt{\eps}}{1-\eps} = \frac{2\dPsi^{1/4}}{3}.
	\]
	W.l.o.g. assume that $\left\Vert p^{(i)}\right\Vert_{2}^{2}\geq\left\Vert p^{(j)}\right\Vert_{2}^{2}$, say $\left\Vert p^{(j)}\right\Vert_{2}=\alpha\cdot\left\Vert p^{(i)}\right\Vert_{2}$
	for some $\alpha\in(0,1]$.
	Then by Lemma \ref{lem_pi2} we have
	$\left\Vert p^{(i)}\right\Vert_{2}^{2} \geq (1-\sqrt{\dPsi}/4)\cdot\left[\min\left\{ \mu(P_{i}),\mu(P_{j})\right\} \right]^{-1}$, and hence
	\begin{eqnarray*}
		\left\Vert p^{(i)}-p^{(j)}\right\Vert_{2}^{2} 
		& = & \left\Vert p^{(i)}\right\Vert_{2}^{2}+\left\Vert p^{(j)}\right\Vert_{2}^{2}-2\left\langle \frac{p^{(i)}}{\left\Vert p^{(i)}\right\Vert_{2}},\frac{p^{(j)}}{\left\Vert p^{(j)}\right\Vert_{2}}\right\rangle \left\Vert p^{(i)}\right\Vert_{2}\left\Vert p^{(j)}\right\Vert_{2}\\
		& \geq & \left(\alpha^{2}-\frac{4\dPsi^{1/4}}{3}\cdot\alpha+1\right)\left\Vert p^{(i)}\right\Vert_{2}^{2}\geq\left[2\cdot\min\left\{ \mu(P_{i}),\mu(P_{j})\right\} \right]^{-1}.
	\end{eqnarray*}
\end{proof}
Using Lemma~\ref{lem_pi2} and Lemma~\ref{lem:pi are well-spread}, the observation that 
$\Upsilon$ can be replaced by $\Psi$ in all statements in~\cite{Peng0Z17} is technically easy.
\bigskip

Our second technical contribution is to show that 
the larger $\ell_2$ distance between estimation centers, 
in Lemma~\ref{lem:pi are well-spread},
strengthens \cite[Lemma 4.5]{Peng0Z17} by a factor of $k$.
Before we state our result, we need some notation.

The normalized Spectral Embedding map $\mF: V \rightarrow \R^k$ 
is defined by
\[
\mF(v)\Def 
\frac{1}{\sqrt{d(v)}}\left(f_{1}(v), \dots, f_{k}(v)\right)^\rot
= \frac{1}{\sqrt{d(v)}}\cdot[Y(v,:)]^{\rot},
\]
for every node $v\in V$.
Recall that the normalized SE $\Yp$ contains duplicate rows,
namely, $d(u)$ many copies of $\mF(u)$ for each node $u \in V$.

Suppose an $\alpha$-approximate $k$-means clustering algorithm outputs
a $k$-way row partition $(R_1,\dots,R_k)$ of $\Yp$.
We can assume w.l.o.g. that all identical rows of $\Yp$ are assigned to same cluster,
and thus $(R_1,\dots,R_k)$ induces a $k$-way node partition $(A_1,\dots,A_k)$ of $G$.
For an arbitrary point set $c_1,\dots,c_k$ in $\R^{k}$, we abuse the notation
and denote the $k$-means cost of a tuple $\{A_{i},c_{i}\}_{i=1}^{k}$ by
\[  
\mathrm{Cost}(\{ A_{i},c_{i}\} _{i=1}^{k}) = 
\sum_{i=1}^{k}\sum_{u\in A_{i}}d(u)\left\Vert \mF(u)-c_i \right\Vert_{2}^{2}.
\]
When each point $c_j=\tfrac{1}{\mu(A_j)}\sum_{u\in A_j}d(u)\mF(u)$ 
is the gravity center of cluster $R_j$,
for brevity we write $\mathrm{Cost}(\{ A_{i} \}_{i=1}^{k})$
to denote the $k$-means cost of tuple $\{A_{i},c_{i}\}_{i=1}^{k}$.

\begin{lem}[Volume Overlap]\label{lem:LW}
	Let $(P_1,\dots,P_k)$ and $(A_1,\dots,A_k)$ be $k$-way node partitions of $G$. 
	Suppose for every permutation $\pi:\{1,\dots,k\}\rightarrow\{1,\dots,k\}$
	there is an index $i\in\{1,\dots,k\}$ such that
	\begin{equation}\label{eq:Perm}
	\mu(A_{i}\triangle P_{\pi(i)})\geq\frac{2\eps}{k}\cdot\mu(P_{\pi(i)}),
	\end{equation}
	where $\eps\in(0,1)$ is a parameter. If $\Psi= 20^{4}\cdot k^3/\dPsi$ for some $\dPsi\in(0,\frac{1}{2}]$, and $\eps\geq 64\APR\cdot k^3/\Psi$ then
	\[
	\mathrm{Cost}(\{ A_{i}\}_{i=1}^{k}) > \frac{2k^2}{\Psi}\APR.
	\]
\end{lem}

With the above lemmas in place, the proof of Part (a) of Theorem~\ref{thm:myPSZ15} 
is then completed as in~\cite{Peng0Z17}. For completeness, we present the proof.

\subsection{Proof of Improved Structural Result}
\label{proof of the lower bound on the cost of A}

In this Section, we prove Part (a.1) of \thmref{myPSZ15}.
Crucial to our analysis is the following result, which we prove in the next
Subsection~\ref{overlineg and f are close}, showing that vectors 
$\widehat{g_i}$ and ${f_i}$ are close, c.f. Figure~\ref{fig:relation}.
\newcommand{\thmOurFgh}
{
	If $\Psi>4\cdot k^{3/2}$, then for every $i\in\{1,\dots,k\}$ 
	the vectors $f_{i}$ and $\widehat{g_i}=\sum_{j=1}^{k}\beta_{j}^{(i)}\overline{g_{j}}$
	satisfy
	\[
		\norm{f_{i}-\widehat{g_{i}}}^{2}\leq\left(1+\frac{3k}{\Psi}\right)\cdot\frac{k}{\Psi}.
	\]
}
\mytheorem{thmOurFgh}{\thmOurFgh}

\begin{lem}[$(P_1,\dots,P_k)$ is a good $k$-means partition]\label{upper bound on the cost of S}
	If $\Psi>4\cdot k^{3/2}$, then there are vectors $\{p^{(i)}\}_{i=1}^{k}$ such that
	\[
	\mathrm{Cost}(\{ P_{i},p^{(i)}\}_{i=1}^{k}) \leq \left(1+\frac{3k}{\Psi}\right)\cdot\frac{k^{2}}{\Psi}.
	\]
\end{lem}

\begin{proof}
	Let $(P_1,\dots,P_k)$ be a $k$-way node partition of $G$ achieving $\hrAvrK$.
	Peng et al.~\cite[Lemma 4.1]{Peng0Z17} showed that
	$\mathrm{Cost}(\{ P_{i},p^{(i)}\}_{i=1}^{k})=\sum_{j=1}^{k}\left\Vert f_{j}-\widehat{g_{j}}\right\Vert_{2}^{2}$, and thus the statement follows by
	\thmref{thmOurFgh}.
	
	For completeness, we now prove the preceding equation.
	By definition, $p_j^{(i)} = \beta_{i}^{(j)}/\sqrt{\mu(P_i)}$ and
	$\widehat{g}_j=\sum_{i=1}^{k}\beta_{i}^{(j)}\cdot \tfrac{D^{1/2}\chi_{P_{i}}}{\sqrt{\mu(P_i)}}$,
	where $\chi_{P_{i}}$ is characteristic vector of the node subset $P_i$.
	Then,
	\begin{eqnarray*}
		& & \mathrm{Cost}(\{P_{i},p^{(i)}\}_{i=1}^{k})=
		\sum_{i=1}^{k}\sum_{u\in P_{i}}d(u)\|\mF(u)-p^{(i)}\|_{2}^{2}\\
		& = &\sum_{j=1}^{k}\sum_{i=1}^{k}\sum_{u\in P_{i}}\Big(f_{j}(u)-\frac{\sqrt{d(u)}}{\sqrt{\mu(P_{i})}}\beta_{i}^{(j)}\Big)^{2}
		= \sum_{j=1}^{k}\left\Vert f_{j}-\widehat{g_{j}}\right\Vert _{2}^{2}.
	\end{eqnarray*}
\end{proof}

\begin{lem}[Only partitions close to $(P_1,\dots,P_k)$ are good]\label{lower bound on the cost of A}
	Under the hypothesis of \thmref{myPSZ15}, the following holds. If for
	every permutation $\sigma:\{1,\dots,k\}\rightarrow\{1,\dots,k\}$ there exists
	an index $i\in\{1,\dots,k\}$ such that
	\begin{equation}\label{eq:LBOHCOA_HYP}
	\mu(A_{i} \triangle P_{\sigma(i)}) \geq \frac{8\APR\dPsi}{10^4k} \cdot\mu(P_{\sigma(i)}).
	\end{equation}
	Then it holds that
	\begin{equation}\label{eq:LBOHCOA_CONL}
	\mathrm{Cost}(\{ A_{i}\} _{i=1}^{k})
	>\frac{2\APR k^2}{\Psi}.
	\end{equation}
\end{lem}

We note that Lemma~\ref{lower bound on the cost of A} follows directly by 
applying Lemma~\ref{lem:LW} with $\eps = 64\cdot\APR\cdot k^3/\Psi$. 
Since $(A_1,\dots,A_k)$ is an $\alpha$ approximate solution to $\tXV{k}$, we obtain a contradiction
\[
	\frac{2\APR k^2}{\Psi} < \mathrm{Cost}(\{ A_{i}\}_{i = 1}^k) 
	\le \APR \cdot \tXV{k} \le \APR \cdot \mathrm{Cost}(\{ P_{i},p^{(i)} \} _{i=1}^{k}) 
	\leq \left(1+\frac{3k}{\Psi}\right)\cdot\frac{\APR k^2}{\Psi}.
\]
Therefore, there exists a permutation $\pi$ (the identity after suitable renumbering of one of the partitions) such that $\mu(A_{i}\triangle P_{i})<\frac{8\APR\dPsi}{10^4k}\cdot\mu(P_i)$ for all $i\in\{1,\dots,k\}$.

Part (a.2)  of \thmref{myPSZ15} follows from Part (a.1). 
Indeed, for $\dprm=8\dPsi/10^4$ we have
\[
	\mu(A_i) \geq  \mu(P_i \cap A_i) = \mu(P_i) - \mu(P_i \setminus A_i) 
	\ge \mu(P_i) - \mu(A_i \triangle P_i) \ge 
	\left(1 - \frac{\APR\dprm}{k}\right)\cdot \mu(P_i)
\]
and $\abs{E(A_i,\overline{A_i})} \le \abs{E(P_i,\overline{P_i})} + \mu(A_i \Delta P_i)$,
since every edge that is counted in $\abs{E(A_i,\overline{A_i})}$ but not in $\abs{E(P_i,\overline{P_i})}$ must have an endpoint in $A_i \Delta P_i$. Thus
\[
	\Phi(A_i) = \frac{\abs{E(A_i,\overline{A_i})}}{\mu(A_i)} \le \frac{\abs{E(P_i,\overline{P_i})} + \frac{\APR\dprm}{k}\cdot\mu(P_i)}{(1 - \frac{\APR\cdot\dprm}{k})\cdot \mu(P_i)} \le \left(1+\frac{2\APR\dprm}{k}\right)\cdot\phi(P_i) + \frac{2\APR\dprm}{k}.
\]
This completes the proof of Part (a) of \thmref{myPSZ15}.

\subsection{Vectors $\widehat{g_i}$ and ${f_i}$ are Close}\label{overlineg and f are close}

In this section, we prove \thmref{thmOurFgh}. 
We argue in a similar manner as in~\cite{Peng0Z17}, 
but in contrast our results depend on the weaken gap parameter $\Psi$.
For completeness, we show in Subsection \ref{subsec:ACM} that the span of the 
first $k$ eigenvectors of $\LG$ equals the span of the projections of $P_i$'s 
characteristic vectors onto the first $k$ eigenvectors.
Then, in Subsection \ref{subsec:AETF}, we conclude the proof of \thmref{thmOurFgh}
by analyzing the eigenvectors $\{f_i\}_{i=1}^{k}$ in terms of projection vectors 
$\{\wfi\}_{i=1}^{k}$.

\subsubsection{Analyzing the Columns of Matrix $F$}\label{subsec:ACM}

We show now that the span of the first $k$ eigenvectors $\{f_{i}\}_{i=1}^{k}$
equals the span of the projection vectors  $\{\wfi\}_{i=1}^{k}$.

\newcommand{\lemSpan}
{
	If $\Psi>k^{3/2}$ then the $\mathrm{span}(\{\wfi\}_{i=1}^{k})=\mathrm{span}(\{f_{i}\}_{i=1}^{k})$ 
	and thus each eigenvector can be expressed as
	$f_{i}=\sum_{j=1}^{k}\beta_{j}^{(i)}\cdot\wfj$ 
	for every $i\in\{1,\dots,k\}$.
}
\mylemma{lemSpan}{\lemSpan}

To prove \lemref{lemSpan}, we build upon the following result established 
by Peng et al.~\cite{Peng0Z17}.

\begin{lem}\cite[Theorem 1.1 Part 1]{Peng0Z17}\label{lem_nrm_gt_fh}
	For $P_i\subset V$ let $\overline{g_{i}}=\frac{D^{1/2}\chi_{P_i}}{\left\Vert D^{1/2}\chi_{P_i}\right\Vert_{2}}$. Then any $i\in\{1,\dots,k\}$ it holds that
	\[
	\left\Vert \overline{g_{i}}-\wfi\right\Vert_{2}^{2}=\sum_{j=k+1}^{n}\left(\alpha_{j}^{(i)}\right)^{2}\leq\frac{\mathcal{R}\left(\overline{g_{i}}\right)}{\lkp}= \frac{\phi(P_{i})}{\lkp}.
	\]
\end{lem}

Our analysis crucially relies on the following two technical lemmas.

\begin{lem}
	\label{lem_fh} For every $i\in\{1,\dots,k\}$ and $p\neq q\in\{1,\dots,k\}$ it holds
	that
	\[
	1-\phi(P_{i})/\lkp\leq\left\Vert \wfi\right\Vert_{2}^{2}=\left\Vert \alpha^{(i)}\right\Vert_{2}^{2}\leq1\quad\text{and}\quad\left|\left\langle \widehat{f_{p}},\widehat{f_{q}}\right\rangle \right|=\left|\left\langle \alpha^{p},\alpha^{q}\right\rangle \right|\leq\frac{\sqrt{\phi(P_{p})\cdot\phi(P_{q})}}{\lkp}.
	\]
\end{lem}
\begin{proof}
	The first part follows by Lemma \ref{lem_nrm_gt_fh} and the following
	chain of inequalities
	\[
	1-\frac{\phi(P_{i})}{\lkp}\leq1-\sum_{j=k+1}^{n}\left(\alpha_{j}^{(i)}\right)^{2}=\left\Vert \wfi\right\Vert_{2}^{2}=\sum_{j=1}^{k}\left(\alpha_{j}^{(i)}\right)^{2}\leq\sum_{j=1}^{n}\left(\alpha_{j}^{(i)}\right)^{2}=1.
	\]
	We show now the second part. Since $\left\{ f_{i}\right\} _{i=1}^{n}$
	are orthonormal eigenvectors we have for all $p\neq q$ that
	\begin{equation}
	\left\langle f_{p},f_{q}\right\rangle =\sum_{l=1}^{n}\alpha_{\ell}^{(p)}\cdot\alpha_{\ell}^{(q)}=0.\label{eq:fpfq}
	\end{equation}
	We combine (\ref{eq:fpfq}) and Cauchy-Schwarz to obtain
	\begin{eqnarray*}
		\left|\left\langle \widehat{f_{p}},\widehat{f_{q}}\right\rangle \right| & = & \left|\sum_{l=1}^{k}\alpha_{\ell}^{(p)}\cdot\alpha_{\ell}^{(q)}\right|=\left|\sum_{l=k+1}^{n}\alpha_{\ell}^{(p)}\cdot\alpha_{\ell}^{(q)}\right|\\
		& \leq & \sqrt{\sum_{l=k+1}^{n}\left(\alpha_{\ell}^{(p)}\right)^{2}}\cdot\sqrt{\sum_{l=k+1}^{n}\left(\alpha_{\ell}^{(q)}\right)^{2}}\leq\frac{\sqrt{\phi(P_{p})\cdot\phi(P_{q})}}{\lkp}.
	\end{eqnarray*}
\end{proof}

\begin{lem}\label{lem_lin_indep}
	If $\Psi> k^{3/2}$
	then the columns $\left\{ \F_{:,i}\right\} _{i=1}^{k}$ are linearly
	independent.
\end{lem}

\begin{proof}
	We show that the columns of matrix $\F$ are almost orthonormal. Consider
	the symmetric matrix $\FT \F$. It is known that $\mathrm{ker}\left(\FT \F\right)=ker(\F)$ and that all eigenvalues of matrix $\FT \F$ are real numbers. We proceeds by showing that
	the smallest eigenvalue $\lambda_{\min}(\FT \F)>0$. This
	would imply that $ker(\F)=\emptyset$ and hence yields the
	statement.
	
	By combining Gersgorin Circle Theorem, Lemma \ref{lem_fh} and Cauchy-Schwarz it holds that
	\begin{eqnarray*}
		\lambda_{\min}(\FT \F) & \geq & \min_{i\in\{1,\dots,k\}}\left\{ \left(\FT \F\right)_{ii}-\sum_{j\neq i}^{k}\left|\left(\FT \F\right)_{ij}\right|\right\} =\min_{i\in\{1,\dots,k\}}\left\{ \left\Vert \alpha^{(i)}\right\Vert_{2}^{2}-\sum_{j\neq i}^{k}\left|\left\langle \alpha^{(j)},\alpha^{(i)}\right\rangle \right|\right\} \\
		& \geq & 1-\sum_{j=1}^{k}\sqrt{\frac{\phi(P_{j})}{\lkp}}\sqrt{\frac{\phi(P_{i^{\star}})}{\lkp}} \geq 1-\sqrt{k}\sqrt{\sum_{j=1}^{k}\frac{\phi(P_{j})}{\lkp}}\sqrt{\frac{\phi(P_{i^{\star}})}{\lkp}} \geq 1-\frac{k^{3/2}}{\Psi} > 0,
	\end{eqnarray*}
	where $i^\star\in\{1,\dots,k\}$ is the index that minimizes the expression above.
\end{proof}

We present now the proof of \lemref{lemSpan}.

\begin{proof}[Proof of \lemref{lemSpan}]
	Let $\nu\in\R^k$ be an arbitrary non-zero vector. Notice that
	\begin{equation}
	\sum_{i=1}^{k}\nu_{i}\cdot\wfi=\sum_{i=1}^{k}\nu_{i}\sum_{j=1}^{k}\alpha_{j}^{(i)}f_{j}=\sum_{j=1}^{k}\left(\sum_{i=1}^{k}\nu_{i}\alpha_{j}^{(i)}\right)f_{j}=\sum_{j=1}^{k}\gamma_{j}f_{j},\quad\text{where}\quad\gamma_{j}=\left\langle \F_{j,:},\nu\right\rangle .\label{eq:gammai}
	\end{equation}
	By Lemma \ref{lem_lin_indep}, the columns $\left\{ \F_{:,i}\right\} _{i=1}^{k}$
	are linearly independent and since $\gamma=\F\nu$, it follows
	that at least one component $\gamma_{j}\neq0$. 
	Hence, the vectors $\{ \wfi \} _{i=1}^{k}$ are linearly independent,
	and since each vector $\wfi$ is a projection onto the span of 
	the first $k$ eigenvectors $\{f_i\}_{i=1}^{k}$, 
	it follows that
	$\mathrm{span}(\{\wfi\}_{i=1}^{k})=\mathrm{span}(\{f_{i}\}_{i=1}^{k})$.
	Thus, each eigenvector $f_i$ can be expressed as a linear combination of
	the projection vectors $\{\wfi\}_{i=1}^{k}$.
\end{proof}

\subsubsection{Analyzing Eigenvectors $f$ in terms of $\wfj$}\label{subsec:AETF}

In this Subsection, we prove \thmref{thmOurFgh}.
Using \lemref{lemSpan}, we first express each eigenvector
$f_{i}=\sum_{j=1}^{k}\beta_{j}^{(i)}\cdot\wfj$ as a linear combination 
of the projection vectors $\{\wfj\}_{j=1}^{k}$,
and we bound the squared $\ell_2$ norm of 
the corresponding coefficient vector $\beta^{(i)}=\B_{:,i}$
for all $i\in\{1,\dots,k\}$.
Then, we conclude the proof of \thmref{thmOurFgh}.

\begin{lem}
	\label{lem_bij}If $\Psi>k^{3/2}$
	then for $i\in\left[k\right]$ it holds
	\[
	\left(1+\frac{2k}{\Psi}\right)^{-1}\leq\sum_{j=1}^{k}\left(\beta_{j}^{(i)}\right)^{2}\leq\left(1-\frac{2k}{\Psi}\right)^{-1}.
	\]
\end{lem}
\begin{proof}
	We show now the upper bound. By \lemref{lemSpan} $f_{i}=\sum_{j=1}^{k}\beta_{j}^{(i)}\wfj$ for all $i\in\{1,\dots,k\}$ and thus
	\begin{eqnarray*}
		1 & = & \left\Vert f_{i}\right\Vert_{2}^{2}=\left\langle \sum_{a=1}^{k}\beta_{a}^{(i)}\widehat{f_{a}},\sum_{b=1}^{k}\beta_{b}^{(i)}\widehat{f_{b}}\right\rangle \\
		& = & \sum_{j=1}^{k}\left(\beta_{j}^{(i)}\right)^{2}\left\Vert \wfj\right\Vert_{2}^{2}+\sum_{a=1}^{k}\sum_{b\neq a}^{k}\beta_{a}^{(i)}\beta_{b}^{(i)}\left\langle \widehat{f_{a}},\widehat{f_{b}}\right\rangle \\
		& \overset{(\star)}{\geq} &
		\left(1-\frac{2k}{\Psi}\right)\cdot\sum_{j=1}^{k}\left(\beta_{j}^{(i)}\right)^{2}.
	\end{eqnarray*}
	To prove the inequality $(\star)$ we consider the two terms separately.
	
	By Lemma \ref{lem_fh}, $\left\Vert \wfj\right\Vert_{2}^{2}\geq1-\phi(P_{j})/\lkp$. We then
	apply $\sum_{i}a_{i}b_{i}\leq(\sum_{i}a_{i})(\sum_{i}b_{i})$
	for all non-negative vectors $a,b$ and obtain
	\[
	\sum_{j=1}^{k}\left(\beta_{j}^{(i)}\right)^{2}\left(1-\frac{\phi(P_{j})}{\lkp}\right)= \sum_{j=1}^{k}\left(\beta_{j}^{(i)}\right)^{2}-\sum_{j=1}^{k}\left(\beta_{j}^{(i)}\right)^{2}\frac{\phi(P_{j})}{\lkp} \geq\left(1-\frac{k}{\Psi}\right)\sum_{j=1}^{k}\left(\beta_{j}^{(i)}\right)^{2}.
	\]
	
	Again by Lemma \ref{lem_fh}, we have $\left|\left\langle \widehat{f_{a}},\widehat{f_{b}}\right\rangle \right|\leq\sqrt{\phi(P_{a})\phi(P_{b})}/\lkp$,
	and by Cauchy-Schwarz it holds
	\begin{eqnarray*}
		\sum_{a=1}^{k}\sum_{b\neq a}^{k}\beta_{a}^{(i)}\beta_{b}^{(i)}\left\langle \widehat{f_{a}},\widehat{f_{b}}\right\rangle
		& \geq & -\sum_{a=1}^{k}\sum_{b\neq a}^{k}\left|\beta_{a}^{(i)}\right|\cdot\left|\beta_{b}^{(i)}\right|\cdot\left|\left\langle \widehat{f_{a}},\widehat{f_{b}}\right\rangle \right|\\
		& \geq & -\frac{1}{\lkp}\sum_{a=1}^{k}\sum_{b\neq a}^{k}\left|\beta_{a}^{(i)}\right|\sqrt{\phi(P_{a})}\cdot\left|\beta_{b}^{(i)}\right|\sqrt{\phi(P_{b})}\\
		& \geq & -\frac{1}{\lkp}\left(\sum_{j=1}^{k}\left|\beta_{j}^{(i)}\right|\sqrt{\phi(P_{j})}\right)^{2}\geq-\frac{k}{\Psi}\cdot\sum_{j=1}^{k}\left(\beta_{j}^{(i)}\right)^{2}.
	\end{eqnarray*}
	The lower bound follows by analogous arguments.
\end{proof}

We are now ready to prove \thmref{thmOurFgh}.

\begin{proof}[Proof of \thmref{thmOurFgh}]
	By \lemref{lemSpan}, we have $f_{i}=\sum_{j=1}^{k}\beta_{j}^{(i)}\wfj$
	and recall that $\widehat{g_{i}}=\sum_{j=1}^{k}\beta_{j}^{(i)}\overline{g_{j}}$
	for all $i\in\{1,\dots,k\}$. 
	Further, by combining triangle inequality, Cauchy-Schwarz,
	Lemma \ref{lem_nrm_gt_fh} and Lemma \ref{lem_bij}, we obtain that
	\begin{eqnarray*}
		&&\left\Vert f_{i}-\widehat{g_{i}}\right\Vert_{2}^{2}=\left\Vert \sum_{j=1}^{k}\beta_{j}^{(i)}\left(\wfj-\overline{g_{j}}\right)\right\Vert_{2}^{2}\leq\left(\sum_{j=1}^{k}\left|\beta_{j}^{i}\right|\cdot\left\Vert \wfj-\overline{g_{j}}\right\Vert_{2}\right)^{2}\\&\leq&\left(\sum_{j=1}^{k}\left(\beta_{j}^{(i)}\right)^{2}\right)\cdot\left(\sum_{j=1}^{k}\left\Vert \wfj-\overline{g_{j}}\right\Vert_{2}^{2}\right)\leq \left(1-\frac{2k}{\Psi}\right)^{-1}\left(\frac{1}{\lkp}\sum_{j=1}^{k}\phi(P_{j})\right)\\ &=&\left(1-\frac{2k}{\Psi}\right)^{-1}\cdot\frac{k}{\Psi} \leq\left(1+\frac{3k}{\Psi}\right)\cdot\frac{k}{\Psi},
	\end{eqnarray*}
	where the last inequality uses $\Psi>4 k$.
\end{proof}

\subsection{Spectral Properties of Matrix $B$}\label{sec:AMB}

In this Section, we prove Theorem~\ref{thmBFC} in two steps.
In Subsection~\ref{subsec:ACSMB}, we analyzes the column space of matrix $\B$ 
and we show that matrix $\BT\B$ is close to the identity matrix.
Then, in Subsection~\ref{subsec:ARIP}, we analyze the row space of matrix $\B$ 
and we prove that matrix $\B\BT$ is close to the identity matrix.

\subsubsection{Analyzing the Column Space of Matrix $B$}\label{subsec:ACSMB}

We show below that the matrix $\BT \B$ is close to the identity matrix.
\begin{lem}\label{lem_Bcol}(Columns)
	If $\Psi>4\cdot k^{3/2}$ then for all distinct $i,j\in\{1,\dots,k\}$
	it holds
	\[
	1-\frac{3k}{\Psi}\leq\left\langle \B_{:,i},\B_{:,i}\right\rangle \leq1+\frac{3k}{\Psi}\quad\text{and}\quad\left|\left\langle \B_{:,i},\B_{:,j}\right\rangle \right|\leq4\sqrt{\frac{k}{\Psi}}.
	\]
	
\end{lem}

\begin{proof}
	By Lemma \ref{lem_bij} it holds that
	\[
	1-\frac{3k}{\Psi}\leq\left\langle \B_{:,i},\B_{:,i}\right\rangle =\sum_{j=1}^{k}\left(\beta_{j}^{(i)}\right)^{2}\leq1+\frac{3k}{\Psi}.
	\]
	Recall that $\widehat{g_{i}}=\sum_{j=1}^{k}\beta_{j}^{(i)}\cdot\overline{g_{j}}$.
	Moreover, since the eigenvectors $\left\{ f_{i}\right\} _{i=1}^{k}$
	and the characteristic vectors $\left\{ \overline{g_{i}}\right\} _{i=1}^{k}$
	are orthonormal by combing Cauchy-Schwarz and by \thmref{thmOurFgh} it holds
	\begin{eqnarray*}
		\left|\left\langle \B_{:,i},\B_{:,j}\right\rangle \right| & = & \sum_{l=1}^{k}\beta_{\ell}^{(i)}\beta_{\ell}^{(j)}=\left\langle \sum_{a=1}^{k}\beta_{a}^{(i)}\cdot\overline{g_{a}},\sum_{b=1}^{k}\beta_{b}^{(j)}\cdot\overline{g_{b}}\right\rangle =\left\langle \widehat{g_{i}},\widehat{g_{j}}\right\rangle \\
		& = & \left\langle \left(\widehat{g_{i}}-f_{i}\right)+f_{i},\left(\widehat{g_{j}}-f_{j}\right)+f_{j}\right\rangle \\
		& = & \left\langle \widehat{g_{i}}-f_{i},\widehat{g_{j}}-f_{j}\right\rangle +\left\langle \widehat{g_{i}}-f_{i},f_{j}\right\rangle +\left\langle f_{i},\widehat{g_{j}}-f_{j}\right\rangle \\
		& \leq & \left\Vert \widehat{g_{i}}-f_{i}\right\Vert_{2}\cdot\left\Vert \widehat{g_{j}}-f_{j}\right\Vert_{2}+\left\Vert \widehat{g_{i}}-f_{i}\right\Vert_{2}+\left\Vert \widehat{g_{j}}-f_{j}\right\Vert_{2}\\
		& \leq & \left(1+\frac{3k}{\Psi}\right)\cdot\frac{k}{\Psi}+2\sqrt{\left(1+\frac{3k}{\Psi}\right)\cdot\frac{k}{\Psi}}\leq4\sqrt{\frac{k}{\Psi}}.
	\end{eqnarray*}
\end{proof}

We demonstrate now that the columns of matrix $\B$ are linearly independent.

\begin{lem}
	\label{lem_linIndB} If $\Psi>25\cdot k^{3}$ then
	the columns $\left\{ \B_{:,i}\right\} _{i=1}^{k}$ are linearly independent.\end{lem}
\begin{proof}
	Since $\mathrm{ker}\left(\B\right)=\mathrm{ker}\left(\B^\rot \B\right)$ and $\B^\rot \B$ is
	$\SPSD$\footnote{We denote by $\SPSD$ the class of symmetric positive semi-definite matrices.} matrix, it suffices to show that the smallest eigenvalue
	\[
	\lambda(\B^\rot \B)=\min_{x\neq0}\frac{x^\rot \B^\rot \B x}{x^\rot x}>0.
	\]
	By Lemma \ref{lem_Bcol},
	\[
	\sum_{i=1}^{k}\sum_{j\neq i}^{k}\left|x_{i}\right|\left|x_{j}\right|\left|\left\langle \beta^{(i)},\beta^{(j)}\right\rangle \right|\leq4\sqrt{\frac{k}{\Psi}}\left(\sum_{i=1}^{k}\left|x_{i}\right|\right)^{2}\leq\left\Vert x\right\Vert_{2}^{2}\cdot4k\sqrt{\frac{k}{\Psi}},
	\]
	and
	\begin{eqnarray*}
		x^\rot \B^\rot \B x & = & \left\langle \sum_{i=1}^{k}x_{i}\beta^{(i)},\sum_{j=1}^{k}x_{j}\beta^{(j)}\right\rangle =\sum_{i=1}^{k}x_{i}^{2}\left\Vert \beta^{(i)}\right\Vert_{2}^{2}+\sum_{i=1}^{k}\sum_{j\neq i}^{k}x_{i}x_{j}\left\langle \beta^{(i)},\beta^{(j)}\right\rangle \\
		& \geq & \left(1-\frac{3k}{\Psi}\right)\left\Vert x\right\Vert_{2}^{2}-\sum_{i=1}^{k}\sum_{j\neq i}^{k}\left|x_{i}\right|\left|x_{j}\right|\left|\left\langle \beta^{(i)},\beta^{(j)}\right\rangle \right|\geq\left(1-5k\sqrt{\frac{k}{\Psi}}\right)\cdot\left\Vert x\right\Vert_{2}^{2}.
	\end{eqnarray*}
	Hence, $\lambda(\B^\rot \B)>0$ and the statement follows.
\end{proof}

\subsubsection{Analyzing the Row Space of Matrix $B$}\label{subsec:ARIP}

In this Subsection, we show that matrix $\B\BT $ is close to the identity matrix. 
We bound now the squared $\ell_2$ norm of the rows in matrix $\B$, 
i.e. the diagonal entries in matrix $\B\BT $.

\begin{lem}
	\label{lem_Brow1}(Rows) If $\Psi\geq 400\cdot k^{3}/\eps^{2}$
	and $\eps\in(0,1)$ then for all distinct $i,j\in\{1,\dots,k\}$ it holds
	\[
	1-\eps\leq\left\langle \B_{i,:},\B_{i,:}\right\rangle \leq1+\eps.
	\]
\end{lem}

\begin{proof}
	We show that the eigenvalues of matrix $\B\BT $ are concentrated around $1$. This would imply that $\chi_{i}^{\rot}\B\BT \chi_{i}=\left\langle \B_{i,:},\B_{i,:}\right\rangle \approx1$, where $\chi_{i}$ is a characteristic vector. By Lemma \ref{lem_Bcol} we have
	\[
	\left(1-\frac{3k}{\Psi}\right)^{2}\leq\left(\beta^{(i)}\right)^{\rot}\cdot \B\B^\rot \cdot\beta^{(i)}=\left\Vert \beta^{(i)}\right\Vert_{2}^{4}+\sum_{j\neq i}^{k}\left\langle \beta^{(j)},\beta^{(i)}\right\rangle ^{2}\leq\left(1+\frac{3k}{\Psi}\right)^{2}+\frac{16k^{2}}{\Psi}\leq1+\frac{23k^{2}}{\Psi}
	\]
	and
	\[
	\left|\left(\beta^{(i)}\right)^{\rot}\cdot \B\B^\rot \cdot\beta^{(j)}\right|\leq\sum_{l=1}^{k}\left|\left\langle \beta^{(i)},\beta^{(l)}\right\rangle \right|\left|\left\langle \beta^{(l)},\beta^{(j)}\right\rangle \right|\leq8\left(1+\frac{3k}{\Psi}\right)\sqrt{\frac{k}{\Psi}}+16\frac{k^{2}}{\Psi} \leq 11 \sqrt{\frac{k}{\Psi}}.
	\]
	By Lemma \ref{lem_linIndB} every vector $x\in\mathbb{R}^{k}$ can
	be expressed as $x=\sum_{i=1}^{k}\gamma_{i}\beta^{(i)}$.
	\begin{eqnarray*}
		x^\rot \B\B^\rot x & = & \sum_{i=1}^{k}\gamma_{i}\left(\beta^{(i)}\right)^{\rot}\cdot \B\B^\rot \cdot\sum_{j=1}^{k}\gamma_{j}\beta^{(j)}\\
		& = & \sum_{i=1}^{k}\gamma_{i}^{2}\left(\beta^{(i)}\right)^{\rot}\cdot \B\B^\rot \cdot\beta^{(i)}+\sum_{i=1}^{k}\sum_{j\neq i}^{k}\gamma_{i}\gamma_{j}\left(\beta^{(i)}\right)^{\rot}\cdot \B\B^\rot \cdot\beta^{(j)}\\
		& \geq & \left(1-\frac{23k^{2}}{\Psi}-11\cdot k\sqrt{\frac{k}{\Psi}}\right)\left\Vert \gamma\right\Vert_{2}^{2}\geq\left(1-14\cdot k\sqrt{\frac{k}{\Psi}}\right)\left\Vert \gamma\right\Vert_{2}^{2}.
	\end{eqnarray*}
	and
	\[
	x^{\rot}x=\sum_{i=1}^{k}\sum_{j=1}^{k}\gamma_{i}\gamma_{j}\left\langle \beta^{(i)},\beta^{(j)}\right\rangle =\sum_{i=1}^{k}\gamma_{i}^{2}\left\Vert \beta^{(i)}\right\Vert_{2}^{2}+\sum_{i=1}^{k}\sum_{j\neq i}^{k}\gamma_{i}\gamma_{j}\left\langle \beta^{(i)},\beta^{(j)}\right\rangle
	\]
	By Lemma \ref{lem_Bcol} we have $\left|\sum_{i=1}^{k}\sum_{j\neq i}^{k}\gamma_{i}\gamma_{j}\left\langle \beta^{(i)},\beta^{(j)}\right\rangle \right|\leq\left\Vert \gamma\right\Vert_{2}^{2}\cdot4k\sqrt{\frac{k}{\Psi}}$
	and $\left\Vert \beta^{(i)}\right\Vert_{2}^{2}\leq1+\frac{3k}{\Psi}$.
	Thus, it holds
	\[
	\left(1-5k\sqrt{\frac{k}{\Psi}}\right)\left\Vert \gamma\right\Vert_{2}^{2}\leq x^\rot x\leq\left(1+5k\sqrt{\frac{k}{\Psi}}\right)\left\Vert \gamma\right\Vert_{2}^{2}.
	\]
	Hence, we have
	\[ 
	1-20k\sqrt{\frac{k}{\Psi}}\leq\lambda(\B\B^\rot)\leq1+20k\sqrt{\frac{k}{\Psi}}.
	\]
\end{proof}

This proves the first part of Theorem~\ref{thmBFC}. 
We turn now to the second part and restate it in the following Lemma.

\begin{lem}\label{lem_Brow2}(Rows)
	If $\Psi\geq 10^4\cdot k^{3}/\eps^{2}$ and $\eps\in(0,1)$ then for all distinct $i,j\in\{1,\dots,k\}$
	it holds
	\[
	\left|\left\langle \B_{i,:},\B_{j,:}\right\rangle \right|\leq\sqrt{\eps}.
	\]
\end{lem}

Let $E \in\R^{k\times k}$ be a symmetric matrix such that 
$\BT \B=\I+E $ and $\left|E _{ij}\right|\leq4\sqrt{k/\Psi}$.
Then,
\begin{equation}\label{eq:defBTB}
	\left(\B\BT \right)^{2}=\B\left(\I+E \right)\BT =\B\BT + \B E \BT.
\end{equation}

We show next that the absolute value of every eigenvalue of matrix 
$\B E \BT$ is small, and further demonstrate that this implies
that all entries of matrix $\B E \BT$ are small.
Then, we conclude the proof of Lemma~\ref{lem_Brow2}.

\begin{lem}\label{lem_BEBT} 
	If $\Psi\geq 40^2\cdot k^{3}/\eps^{2}$ and $\eps\in(0,1)$, 
	then every eigenvalue $\lambda$ of matrix $\B E \BT$ satisfies
	\[
	\left|\lambda(\B E \BT )\right| \leq \eps/5.
	\]
\end{lem}

\begin{proof}
	Let $z=\BT x$. We upper bound the quadratic form
	\[
	\left|x^{\rot}\B E \BT x\right| = \left|z^{\rot}E z\right|\leq \sum_{ij}\left|E _{ij}\right|\left|z_{i}\right|\left|z_{j}\right|\leq 4\sqrt{\frac{k}{\Psi}}\cdot\left(\sum_{i=1}^{k}\left|z_{i}\right|\right)^{2}\leq \left\Vert z\right\Vert_{2}^{2}\cdot4k\sqrt{\frac{k}{\Psi}}.
	\]
	By Lemma~\ref{lem_Brow1}, we have $1-\eps\leq\lambda(\B\B^\rot)\leq1+\eps$
	and since 
	$\left\Vert z\right\Vert_{2}^{2}=\frac{x\B\BT x}{x^{\rot}x}\cdot\left\Vert x\right\Vert_{2}^{2}$,
	it follows that
	\[
	\frac{\left\Vert z\right\Vert_{2}^{2}}{1+\eps}\leq\left\Vert x\right\Vert_{2}^{2}\leq\frac{\left\Vert z\right\Vert_{2}^{2}}{1-\eps},
	\]
	and hence
	\[
	\left|\lambda(\B E \BT )\right|\leq\max_{x}\frac{\left|x^{\rot}\B E \BT x\right|}{x^{\rot}x}\leq4\left(1+\eps\right)\cdot k\sqrt{\frac{k}{\Psi}}\leq\eps/5.
	\]
\end{proof}

\begin{lem}\label{lem_absBEBT}
	If $\Psi\geq40^2\cdot k^{3}/\eps^{2}$ and $\eps\in(0,1)$,
	then it holds that $\lvert(\B E \B^\rot)_{ij}\rvert\leq\eps/5$
	for every $i,j\in\{1,\dots,k\}$.
\end{lem}
\begin{proof}
	Since matrix $E \in\R^{k\times k}$ is a symmetric, 
	by construction matrix, $\B E \B^\rot\in\R^{k\times k}$ is also symmetric.
	Using the SVD Theorem, there is an orthonormal basis $\left\{ u_{i}\right\} _{i=1}^{k}$
	such that $\B E \B^\rot =\sum_{i=1}^{k}\lambda_{i}(\B E \B^\rot)\cdot u_{i}u_{i}^\rot $.
	Thus, it suffices to bound the expression
	\[
	\lvert(\B E \B^\rot)_{ij}\rvert \leq \sum_{l=1}^{k}\lvert\lambda_{\ell}(\B E \B^\rot)\rvert \cdot\lvert(u_{\ell}u_{\ell}^\rot)_{ij}\rvert.
	\]
	Let $\U\in\R^{k\times k}$ be a square matrix whose $i$-th column is vector $u_i$.
	By construction, matrix $\U$ is orthogonal and satisfies $\UT \U = \I = \U\UT$. 
	In particular, it holds that $\left\Vert \U_{i,:}\right\Vert_2^2=1$, for all $i$. 
	Therefore, we have
	\[
	\sum_{l=1}^{k}\lvert(u_{\ell})_{i}\rvert\cdot\lvert(u_{\ell})_{j}\rvert\leq\sqrt{\left\Vert \U_{i,:}\right\Vert_2^{2}}\sqrt{\left\Vert \U_{j,:}\right\Vert_2^{2}}=1.
	\]
	We apply now Lemma \ref{lem_BEBT} to obtain
	\[
	\sum_{l=1}^{k}\lvert\lambda_{\ell}(\B E \B^\rot)\rvert \cdot \lvert(u_{\ell}u_{\ell}^\rot)_{ij}\rvert\leq \frac{\eps}{5}\cdot\sum_{l=1}^{k}\lvert(u_{\ell})_{i}\rvert\cdot\lvert(u_{\ell})_{j}\rvert \leq \frac{\eps}{5}.\hfill\qedhere
	\]	
\end{proof}

We are now ready to prove Lemma \ref{lem_Brow2}.

\begin{proof}[Proof of Lemma~\ref{lem_Brow2}]
	By \eqref{eq:defBTB} we have $\left(\B\BT \right)^{2}=\B\BT +\B E \BT $. Observe that the $(i,j)$-th entry of matrix $\B\BT $ is equal to the inner product between the $i$-th and $j$-th row of matrix $\B$, i.e. $\left(\B\BT \right)_{ij}=\left\langle \B_{i,:},\B_{j,:}\right\rangle$. Moreover, we have
	\begin{eqnarray*}
		\left[\left(\B\BT \right)^{2}\right]_{ij} & = & \sum_{l=1}^{k}\left(\B\BT \right)_{i,l}\left(\B\BT \right)_{l,j}=\sum_{l=1}^{k}\left\langle \B_{i,:},\B_{l,:}\right\rangle \left\langle \B_{l,:},\B_{j,:}\right\rangle.
	\end{eqnarray*}
	For the entries on the main diagonal, it holds
	\[
	\left\langle \B_{i,:},\B_{i,:}\right\rangle ^{2}+\sum_{l\neq i}^{k}\left\langle \B_{i,:},\B_{l,:}\right\rangle ^{2} = [(\B\BT )^{2}]_{ii} = [\B\BT +\B E \BT ]_{ii} =
	\left\langle \B_{i,:},\B_{i,:}\right\rangle +\left(\B E \BT \right)_{ii},
	\]
	and hence by applying Lemma~\ref{lem_Brow1} with $\eps^{\prm}=\eps/5$ and Lemma~\ref{lem_absBEBT} with $\eps^{\prm}=\eps$ we obtain
	\[\left\langle \B_{i,:},\B_{j,:}\right\rangle ^{2}\leq\sum_{l\neq i}\left\langle \B_{i,:},\B_{l,:}\right\rangle ^{2} \le \left(1 + \frac{\eps}{5}\right) + \frac{\eps}{5} - \left(1 - \frac{\eps}{5}\right)^{2} \leq \eps.
	\]
\end{proof}

\subsection{Volume Overlap Lemma}\label{sec:lemLW}

In this section, we prove Lemma~\ref{lem:LW}. 
Our main technical contribution is to strengthen the lower bound of 
$k$-means cost in \cite[Lemma 4.5]{Peng0Z17} by a factor of $k$,
under the weaken gap assumption.
	
We begin by stating a useful Corollary of Lemma~\ref{lem:pi are well-spread}.

\begin{cor}\label{cor:pi are well-spread}
	Let $\Psi= 20^{4}\cdot k^3/\dPsi$ for some $\dPsi\in(0,1/2]$. Suppose $c_i$ is the center of a cluster $A_i$. If $\left\Vert c_{i}-p^{(i_{1})}\right\Vert_{2}\geq\left\Vert c_{i}-p^{(i_{2})}\right\Vert_{2}$, then it holds that
	\[
	\left\Vert c_{i}-p^{(i_{1})}\right\Vert_{2}^{2}\geq\frac{1}{4}\left\Vert p^{(i_{1})}-p^{(i_{2})}\right\Vert_{2}^{2}\geq\left[8\cdot\min\left\{ \mu(P_{i_{1}}),\mu(P_{i_{2}})\right\} \right]^{-1}.
	\]
\end{cor}

We restate now \cite[Lemma 4.6]{Peng0Z17} whose analysis crucially relies on 
the following function
\begin{equation}\label{eq:funcSigma}
\sigma(\ell)=\arg\max_{j\in\{1,\dots,k\}}\frac{\mu(A_{\ell}\cap P_{j})}{\mu(P_{j})},
\quad\text{for all }\ell\in\{1,\dots,k\}.
\end{equation}

\begin{lem}\cite[Lemma 4.6]{Peng0Z17}\label{lemFixed}
	Let $(P_1,\dots,P_k)$ and $(A_1,\dots,A_k)$ be $k$-way node partitions of $G$.
	Suppose for every permutation $\pi:\{1,\dots,k\}\rightarrow\{1,\dots,k\}$
	there is an index $i\in\{1,\dots,k\}$ such that
	\begin{equation}\label{eq:Perm}
	\mu(A_{i}\triangle P_{\pi(i)})\geq2\eps\cdot\mu(P_{\pi(i)}),
	\end{equation}
	where $\eps\in(0,1/2)$ is a parameter. Then one of the following three statements holds:\\
	1. If $\sigma$ is a permutation and $\mu(P_{\sigma(i)}\backslash A_i)\geq\eps\cdot\mu(P_{\sigma(i)})$, then for every index $j\neq i$ there is a real $\eps_{j}\geq 0$ such that
	\[
	\mu(A_{j}\cap P_{\sigma(j)})\geq\mu(A_{j}\cap P_{\sigma(i)})\geq\eps_{j}\cdot\min\{\mu(P_{\sigma(j)}),\mu(P_{\sigma(i)})\},
	\]
	and $\sum_{j\neq i}\eps_{j}\geq\eps$.\\
	2. If $\sigma$ is a permutation and $\mu(A_i\backslash P_{\sigma(i)})\geq\eps\cdot\mu(P_{\sigma(i)})$, then for every $j\neq i$ there is a real $\eps_j\geq 0$ such that
	\[
	\mu(A_i\cap P_{\sigma(i)})\geq\eps_{j}\cdot\mu(P_{\sigma(i)}),\quad \mu(A_i\cap P_{\sigma(j)})\geq\eps_{j}\cdot\mu(P_{\sigma(i)}),
	\]
	and $\sum_{j\neq i}\eps_{j}\geq\eps$.\\
	3. If $\sigma$ is not a permutation, then there is an index $\ell\not\in\{\sigma(1),\dots,\sigma(k)\}$ and for every index $j$ there is a real $\eps_{j}\geq 0$ such that
	\[
	\mu(A_j\cap P_{\sigma(j)})\geq\mu(A_j\cap P_{\ell}) \geq \eps_{j}\cdot\min\{\mu(P_{\sigma(j)}),\mu(P_{\ell})\},
	\]
	and $\sum_{j=1}^{k}\eps_{j}=1$.
\end{lem}

We strengthen now the lower bound of $k$-means cost in \cite[Lemma 4.5]{Peng0Z17} by a factor of $k$.

\begin{lem}\label{lem:LWeps}
	Suppose the hypothesis of Lemma~\ref{lemFixed} is satisfied and $\Psi= 20^{4}\cdot k^3/\dPsi$ for some $\dPsi\in(0,1/2]$. Then it holds
	\[
	\mathrm{Cost}(\{A_{i},c_{i}\}_{i=1}^{k}) \geq \frac{\eps}{16} - \frac{2k^{2}}{\Psi}.
	\]
\end{lem}

\begin{proof}
	By definition
	\begin{equation}
		\mathrm{Cost}(\{A_{i},c_{i}\}_{i=1}^{k}) = 
		\sum_{i=1}^{k}\sum_{j=1}^{k}\sum_{u\in A_{i}\cap P_{j}}d(u)
		\left\Vert \mF(u)-c_{i}\right\Vert_{2}^{2}\Def\Lambda.
	\end{equation}
	Since for every vectors $x,y,z\in\mathbb{R}^{k}$ it holds
	\[
		2\left(\left\Vert x-y\right\Vert_{2}^{2}+\left\Vert z-y\right\Vert_{2}^{2}\right)\geq\left(\left\Vert x-y\right\Vert_{2}+\left\Vert z-y\right\Vert_{2}\right)^{2}\geq\left\Vert x-z\right\Vert_{2}^{2},
	\]
	we have for all indices $i,j\in\{1,\dots,k\}$ that
	\begin{equation}\label{eq:LBPC}
		\left\Vert \mF(u)-c_{i}\right\Vert_{2}^{2}\geq\frac{\left\Vert p^{(j)}-c_{i}\right\Vert_{2}^{2}}{2}-\left\Vert \mF(u)-p^{(j)}\right\Vert_{2}^{2}.
	\end{equation}
	Our proof proceeds by considering three cases. Let $i\in\{1,\dots,k\}$ be the index from the hypothesis in Lemma~\ref{lemFixed}.\\
	
	\textbf{Case 1.} Suppose the first conclusion of Lemma~\ref{lemFixed} holds. For every index $j\neq i$ let
	\[
		p^{\gamma(j)}=\begin{cases}
		p^{\sigma(j)} & \text{, if }\left\Vert p^{\sigma(j)}-c_{j}\right\Vert_{2}\geq\left\Vert p^{\sigma(i)}-c_{j}\right\Vert_{2}; \\
		p^{\sigma(i)} & \text{, otherwise}.
		\end{cases}
	\]
	Then by combining (\ref{eq:LBPC}), Corollary~\ref{cor:pi are well-spread} and Lemma~\ref{upper bound on the cost of S}, we have
	\begin{eqnarray*}
		\Lambda &\geq& \frac{1}{2}\sum_{j\ne i}\sum_{u\in A_{j}\cap P_{\gamma(j)}}d(u)\left\Vert p^{\gamma(j)}-c_{j}\right\Vert_{2}^{2}-\sum_{j\ne i}\sum_{u\in A_{j}\cap P_{\gamma(j)}}\left\Vert \mF(u)-p^{\gamma(j)}\right\Vert_{2}^{2}\\
		& \geq & \frac{1}{16}\sum_{j\ne i}\frac{\mu(A_{j}\cap P_{\gamma(j)})}{\min\{\mu(P_{\sigma(i)}),\mu(P_{\sigma(j)})\}}- \left(1+\frac{3k}{\Psi}\right)\cdot\frac{k^{2}}{\Psi} \geq\frac{\eps}{16}-\frac{2k^{2}}{\Psi}.
	\end{eqnarray*}
	
	\textbf{Case 2.} Suppose the second conclusion of Lemma~\ref{lemFixed} holds. 
	Notice that if $\mu(A_i\cap P_{\sigma(i)})\leq(1-\eps)\cdot\mu(P_{\sigma(i)})$ 
	then $\mu(P_{\sigma(i)}\backslash A_i)\geq\eps\cdot\mu(P_{\sigma(i)})$ and 
	thus we can argue as in Case 1. Hence, we can assume that it holds
	\begin{equation}\label{eq:goodCase}
	\mu(A_i\cap P_{\sigma(i)})\geq(1-\eps)\cdot\mu(P_{\sigma(i)}).
	\end{equation}
	We proceed by analyzing two subcases.

	a) If $\left\Vert p^{\sigma(j)}-c_{i}\right\Vert_{2}\geq\left\Vert p^{\sigma(i)}-c_{i}\right\Vert$ holds for all $j\neq i$ then by combining (\ref{eq:LBPC}), Corollary~\ref{cor:pi are well-spread} and Lemma~\ref{upper bound on the cost of S} it follows
	\begin{eqnarray*}
		\Lambda &\geq& \frac{1}{2}\sum_{j\neq i}\sum_{u\in A_{i}\cap P_{\sigma(j)}}d(u)\left\Vert p^{\sigma(j)}-c_{i}\right\Vert_{2}^{2}-\sum_{j\neq i}\sum_{u\in A_{i}\cap P_{\sigma(j)}}\left\Vert \mF(u)-p^{\sigma(j)}\right\Vert_{2}^{2}\\
		& \geq & \frac{1}{2}\sum_{j\neq i}\frac{\mu(A_{i}\cap P_{\sigma(j)})}{\min\{\mu(P_{\sigma(i)}),\mu(P_{\sigma(j)})\}} - \left(1+\frac{3k}{\Psi}\right)\cdot\frac{k^{2}}{\Psi} \geq\frac{\eps}{16}-\frac{2k^{2}}{\Psi}.
	\end{eqnarray*}
	
	b) Suppose there is an index $j\neq i$ such that $\left\Vert p^{\sigma(j)}-c_{i}\right\Vert_{2}<\left\Vert p^{\sigma(i)}-c_{i}\right\Vert$. Then by triangle inequality combined with Corollary~\ref{cor:pi are well-spread} we have
	\[
		\left\Vert p^{\sigma(i)}-c_{i}\right\Vert_{2}^{2} \geq \frac{1}{4}\left\Vert p^{\sigma(i)}-p^{\sigma(j)}\right\Vert_{2}\geq \left[8\cdot\min\{\mu(P_{\sigma(i)}),\mu(P_{\sigma(j)})\}\right]^{-1}.
	\]
	Thus, by combining (\ref{eq:LBPC}), (\ref{eq:goodCase}) and Lemma~\ref{upper bound on the cost of S} we obtain
	\begin{eqnarray*}
		\Lambda &\geq& \frac{1}{2}\sum_{u\in A_{i}\cap P_{\sigma(i)}}d(u)\left\Vert p^{\sigma(i)}-c_{i}\right\Vert_{2}^{2}-\sum_{u\in A_{i}\cap P_{\sigma(i)}}d(u)\left\Vert \mF(u)-p^{\sigma(i)}\right\Vert_{2}^{2}\\
		&\geq& \frac{1}{16}\cdot\frac{\mu(A_{i}\cap P_{\sigma(i)})}{\min\{\mu(P_{\sigma(i)}),\mu(P_{\sigma(j)})\}} - \left(1+\frac{3k}{\Psi}\right)\cdot\frac{k^{2}}{\Psi} \geq \frac{1-\eps}{16}-\frac{2k^{2}}{\Psi}.
	\end{eqnarray*}
	
	\textbf{Case 3.} Suppose the third conclusion of Lemma~\ref{lemFixed} holds, i.e., $\sigma$ is not a permutation. Then there is an index $\ell\in\{1,\dots,k\}\setminus\{ \sigma(1),\dots,\sigma(k)\}$ and for every index $j\in\{1,\dots,k\}$ let
	\[
	p^{\gamma(j)}=\begin{cases}
	p^{\ell} & \text{, if }\left\Vert p^{\ell}-c_{j}\right\Vert_{2}\geq\left\Vert p^{\sigma(j)}-c_{j}\right\Vert_{2}; \\
	p^{\sigma(j)} & \text{, otherwise}.
	\end{cases}
	\]
	By combining (\ref{eq:LBPC}), Corollary~\ref{cor:pi are well-spread} and Lemma~\ref{upper bound on the cost of S} it follows that
	\begin{eqnarray*}
		\Lambda&\geq&\frac{1}{2}\sum_{j=1}^{k}\sum_{u\in A_{j}\cap P_{\gamma(j)}}d(u)\left\Vert p^{\gamma(j)}-c_{j}\right\Vert_{2}^{2}-\sum_{j=1}^{k}\sum_{u\in A_{j}\cap P_{\gamma(j)}}d(u)\left\Vert \mF(u)-p^{\gamma(j)}\right\Vert_{2}^{2}\\
		& \geq & \frac{1}{16}\sum_{j=1}^{k}\frac{\mu(A_{j}\cap P_{\gamma(j)})}{\min\{\mu(P_{\sigma(j)}),\mu(P_{\ell})\}} - \left(1+\frac{3k}{\Psi}\right)\cdot\frac{k^{2}}{\Psi} \geq \frac{1}{16}-\frac{2k^{2}}{\Psi}.
	\end{eqnarray*}
\end{proof}

We are now ready to prove Lemma~\ref{lem:LW}.

\begin{proof}[Proof of Lemma~\ref{lem:LW}]
	We apply Lemma~\ref{lemFixed} with $\eps^\prm=\eps/k$. Then, by Lemma~\ref{lem:LWeps} we have
	\[
	\mathrm{Cost}(\{A_{i},c_{i}\}_{i=1}^{k}) \geq \frac{\eps}{16k} - \frac{2k^{2}}{\Psi},
	\]
	and the desired result follows by setting $\eps\geq 64\APR\cdot k^3/\Psi$.
\end{proof}

\section{Analysis of Approximate Spectral Clustering}
\label{sec:AnalysisOfApproximateSpectralClustering}

\subsection{Normalized Spectral Embedding}\label{sec:OrgSEisEpsSep}

In this Subsection, we prove \thmref{thmGapTriK}, showing that 
the normalized SE $\Yp$ is $\eps$-separated.
For convenience of the reader, we restate the result.\\

\againtheorem{thmGapTriK}
{
	Let $G$ be a graph that satisfies $\Psi=20^{4}\cdot k^{3}/\dPsi$,
	$\dPsi\in(0,1/2]$ and $k/\dPsi\geq10^{9}$. Then for $\eps=6\cdot10^{-7}$ it holds
	\begin{equation}\label{eq:tXVk}
	\tXV{k}\leq\eps^{2}\cdot \tXV{k-1}.
	\end{equation}	
}

We establish first a lower bound on $\tXV{k-1}$.
\begin{lem}\label{lem_LBtkx}
	Let $G$ be a graph that satisfies $\Psi=20^{4}\cdot k^{3}/\dPsi$ for some $\dPsi\in(0,1/2]$. Then for $\dPsi^{\prime}=2\dPsi/20^{4}$ it holds
	\begin{equation}\label{eq:cmpLBTKX}
	\tXV{k-1}\geq\frac{1}{12}-\frac{\dPsi^{\prime}}{k}.
	\end{equation}
\end{lem}
Before we present the proof of Lemma~\ref{lem_LBtkx}, we show that it implies (\ref{eq:tXVk}). 
By Lemma~\ref{upper bound on the cost of S}, we have
\[
\tXV{k}\leq\frac{2k^{2}}{\Psi}=\frac{\dPsi^{\prime}}{k},
\]
and thus, by applying Lemma~\ref{lem_LBtkx} with $k/\dPsi\geq10^{9}$ and $\eps=6\cdot10^{-7}$,
we obtain
\[
\tXV{k-1} \geq \frac{1}{12}-\frac{\dPsi^{\prm}}{k} = \frac{1}{12}- \frac{2}{20^{4}}\cdot\frac{\dPsi}{k} \geq \frac{10^{10}}{9\cdot2^5}   \cdot\frac{\dPsi}{k} = \frac{1}{\eps^{2}}\cdot\frac{\dPsi^{\prm}}{k} \geq \frac{1}{\eps^{2}}\cdot\tXV{k}.
\]

\subsubsection{Proof of Lemma~\ref{lem_LBtkx}}

We argue in a similar manner as in Lemma~\ref{lem:LWeps} (c.f.  Case 3). We start by giving some notation, then we establish Lemma~\ref{lem_muIneq} and apply it in the proof of Lemma~\ref{lem_LBtkx}.

We redefine the function $\sigma$, see \eqref{eq:funcSigma}, such that for any two partitions $(P_{1},\dots,P_{k})$ and $(Z_{1},\dots,Z_{k-1})$ of $V$, we define a mapping $\sigma:\{1,\dots,k-1\}\mapsto\{1,\dots,k\}$ by
\[
\sigma(i)=\arg\max_{j\in\{1,\dots,k\}}\frac{\mu(Z_{i}\cap P_{j})}{\mu(P_{j})},\quad\text{for every }i\in\{1,\dots,k-1\}.
\]

We lower bound now the overlapping of clusters between any $k$-way and 
$(k-1)$-way partitions of $V$ in terms of volume.

\begin{lem}\label{lem_muIneq}
	Suppose $(P_{1},\dots,P_{k})$ and $(Z_{1},\dots,Z_{k-1})$ are partitions of $V$. Then for any index $\ell\in\{1,\dots,k\}\setminus\{ \sigma(1),\dots,\sigma(k-1)\}$ (there is at least one such $\ell$) and for every $i\in\{1,\dots,k-1\}$ it holds
	\[
	\left\{ \mu(Z_{i}\cap P_{\sigma(i)}),\mu(Z_{i}\cap P_{\ell})\right\} \geq\tau_{i}\cdot\min\left\{ \mu(P_{\ell}),\mu(P_{\sigma(i)})\right\},
	\]
	where $\sum_{i=1}^{k-1}\tau_{i}=1$ and $\tau_{i}\geq0$.
\end{lem}

\begin{proof}
	By pigeonhole principle there is an index $\ell\in\{1,\dots,k\}$
	such that $\ell\notin\left\{ \sigma(1),\dots,\sigma(k-1)\right\}$. Thus, for every $i\in\{1,\dots,k-1\}$ we have $\sigma(i)\neq\ell$ and
	\[
	\frac{\mu(Z_{i}\cap P_{\sigma(i)})}{\mu(P_{\sigma(i)})}\geq\frac{\mu(Z_{i}\cap P_{\ell})}{\mu(P_{\ell})}\Def\tau_{i},
	\]
	where $\sum_{i=1}^{k-1}\tau_{i}=1$ and $\tau_{i}\geq0$ for all $i$. Hence, the statement follows.
\end{proof}

We present now the proof of Lemma~\ref{lem_LBtkx}.

\begin{proof}[Proof of Lemma~\ref{lem_LBtkx}]
	Let $(Z_{1},\dots,Z_{k-1})$ be a $(k-1)$-way partition of $V$ with centers $c_{1}^{\prime},\dots, c_{k-1}^{\prime}$ that achieves $\tXV{k-1}$, and $(P_{1},\dots,P_{k})$ be a $k$-way partition of $V$ achieving $\hrAvrK$. Our goal now is to lower bound the optimum $(k-1)$-means cost
	\begin{equation}
	\tXV{k-1}=\sum_{i=1}^{k-1}\sum_{j=1}^{k}\sum_{u\in Z_{i}\cap P_{j}}d_{u}\left\Vert \mF(u)-c_{i}^{\prime}\right\Vert_{2}^{2}.\label{eq:tri_k1}
	\end{equation}
	By Lemma~\ref{lem_muIneq} there is an index $\ell\in\{1,\dots,k\}\setminus\{ \sigma(1),\dots,\sigma(k-1)\}$. For $i\in\{1,\dots,k-1\}$ let
	\[
	p^{\gamma(i)}=\begin{cases}
	p^{\ell} & \text{, if }\left\Vert p^{\ell}-c_{i}^{\prime}\right\Vert_{2}\geq\left\Vert p^{\sigma(i)}-c_{i}^{\prime}\right\Vert_{2}; \\
	p^{\sigma(i)} & \text{, otherwise}.
	\end{cases}
	\]
	Then by combining Corollary~\ref{cor:pi are well-spread} and Lemma~\ref{lem_muIneq}, we have
	\begin{equation}\label{eq:FucLB}
	\left\Vert p^{\gamma(i)}-c_{i}^{\prime}\right\Vert_{2}^{2} \geq \left[8\cdot\min\left\{ \mu(P_{\ell}),\mu(P_{\sigma(i)})\right\} \right]^{-1}\,\,\text{and}\,\,\, \mu(Z_{i}\cap P_{\gamma(i)})\geq\tau_{i}\cdot\min\left\{ \mu(P_{\ell}),\mu(P_{\sigma(i)})\right\},
	\end{equation}
	where $\sum_{i=1}^{k-1}\tau_{i}=1$. We now lower bound the expression
	in (\ref{eq:tri_k1}). Since
	\[
	\left\Vert \mF(u)-c_{i}^{\prime}\right\Vert_{2}^{2} \geq \frac{1}{2}\left\Vert p^{\gamma(i)}-c_{i}^{\prime}\right\Vert_{2}^{2}-\left\Vert \mF(u)-p^{\gamma(i)}\right\Vert_{2}^{2},
	\]
	it follows for $\dPsi^{\prm}=2\dPsi/20^{4}$ that
	\begin{eqnarray*}
		\triangle_{k-1}(\cX_{V}) & = & \sum_{i=1}^{k-1}\sum_{j=1}^{k}\sum_{u\in Z_{i}\cap P_{j}}d_{u}\left\Vert \mF(u)-c_{i}^{\prime}\right\Vert_{2}^{2}
		\geq  \sum_{i=1}^{k-1}\sum_{u\in Z_{i}\cap P_{\gamma(i)}}d_{u}\left\Vert \mF(u)-c_{i}^{\prime}\right\Vert_{2}^{2}\\
		& \geq & \frac{1}{2}\sum_{i=1}^{k-1}\sum_{u\in Z_{i}\cap P_{\gamma(i)}}d_{u}\left\Vert p^{\gamma(i)}-c_{i}^{\prime}\right\Vert_{2}^{2}-\sum_{i=1}^{k-1}\sum_{u\in Z_{i}\cap P_{\gamma(i)}}d_{u}\left\Vert \mF(u)-p^{\gamma(i)}\right\Vert_{2}^{2}\\
		& \geq & \frac{1}{2}\sum_{i=1}^{k-1}\frac{\mu(Z_{i}\cap P_{\gamma(i)})}{8\cdot\min\left\{ \mu(P_{\gamma(i)}),\mu(P_{\sigma(i)})\right\} }-\sum_{i=1}^{k}\sum_{u\in P_{i}}d_{u}\left\Vert \mF(u)-p^{i}\right\Vert_{2}^{2}\\
		& \geq & \frac{1}{16}-\frac{\dPsi^{\prime}}{k},
	\end{eqnarray*}
	where the last inequality holds due to (\ref{eq:FucLB}) and Lemma~\ref{upper bound on the cost of S}.
\end{proof}

\subsection{Approximate Normalized Spectral Embedding}\label{sec:ProofOfPartII}

In this Subsection, we prove \thmref{thmEasySpectralEmbedding}, which shows that
the approximate normalized SE $\wYp$, computed via the Power method,
is $\eps$-separated.

Before we state our results, we need some notation.
Let $\XpOpt$ be an \emph{indicator} matrix, see~\eqref{eq:ZXXTZeqCost},
corresponding to an optimal $k$-way row partition of the normalized SE $\Yp$.	
Then, the optimum $k$-means cost of $\Yp$ in matrix notation reads
\[
	\tXV{k} = \lVert \Yp - \XpOpt(\XpOpt)^{\rot}\Yp \rVert _{F}^{2}.
\]
Similarly, for the approximate normalized SE $\wYp$,
let $\wXpOpt$ be an indicator matrix such that
\[
	\twXVs{k} = \lVert \wYp - \wXpOpt(\wXpOpt)^{\rot} \wYp \rVert _{F}^{2}.
\]

In Subsection~\ref{subsec:lemMyPWM}, using techniques from \cite[Lemma 5]{BoutsidisKG15} 
and \cite[Lemma 7]{BM14}, we prove the following statement.

\begin{lem}\label{my_lem_PWM_FN}
	Let $\lambda_{k}$ and $\lambda_{k+1}$ be the $k$-th 
	and $(k+1)$-st smallest eigenvalue of $\LG$, 
	$Y$ be the canonical SE,
	and $S\in\R^{n\times k}$ be a matrix whose entries are i.i.d. samples 
	from the standard Gaussian distribution.
	For any $\beta,\epsilon\in(0,1)$ and
	$p\geq\ln(8nk/\epsilon\beta)\big/\ln(1/\gamma_{k})$,
	where $\gamma_{k}=\frac{2-\lambda_{k+1}}{2-\lk }<1$,
	compute the approximate SE $\wY$ via the Power method:
	\[
	1)\,\, M \Def I+D^{-1/2}AD^{-1/2} ;\quad 2)\,\,
	\text{Let }\wU \widetilde{\Sigma}\widetilde{V}^{\rot}\text{ be the SVD of }M^{p}S; \quad\text{and}\quad 3)\,\, \wY \Def \wU \in\R^{n\times k}.
	\]
	Then, with probability at least $1-2e^{-2n}-3\beta$, it holds that
	\[
	\lVert YY^{\rot} - \wY \wY^{\rot} \rVert_{F}\leq\epsilon.
	\]
\end{lem}
\bigskip

In Subsection~\ref{subsec:thmMyPWM}, we establish technical lemmas that 
allows us to apply the proof technique developed in~\cite[Theorem 6]{BoutsidisKG15}
for approximate SE $\wY$, and to prove a similar statement 
for the approximate normalized SE $\wYp$.

\smallskip\smallskip
\againtheorem{thmMyPWM}
{
	Let $\eps,\delta_p\in(0,1)$ be arbitrary.
	Compute the approximate normalized SE $\wYp$ via the Power method
	with $p\geq\ln(8nk/\epsilon\delta_p)\big/\ln(1/\gamma_{k})$ iterations
	and $\gamma_{k}=(2-\lambda_{k+1})/(2 -\lk)<1$.
	Run on the rows of $\wYp$ an $\alpha$-approximate $k$-means clustering algorithm 
	with failure probability $\delta_{\alpha}$.
	Let the outcome be a clustering indicator matrix $\wXaP\in\R^{n\times k}$.
	Then, with probability at least $1 - 2e^{-2n} - 3\delta_p-\delta_{\alpha}$,
	it holds that
	\[
	\left\Vert \Yp -\wXaP\left(\wXaP\right)^{\rot}\Yp \right\Vert_{F}^{2}\leq(1+4\eps)\cdot\alpha\cdot\left\Vert \Yp -\XpOpt\left(\XpOpt\right)^{\rot}\Yp \right\Vert_{F}^{2}+4\eps^{2}.
	\]
}

In Subsection~\ref{subsec:ThmPartTwo}, we prove \thmref{thmEasySpectralEmbedding} using 
Lemma~\ref{my_lem_PWM_FN} and \thmref{thmMyPWM}, showing that
the approximate normalized SE $\wYp$, computed via the Power method,
is $\eps$-separated.

\smallskip\smallskip
\againtheorem{thmEasySpectralEmbedding}
{
	Assume $\Psi=20^{4}\cdot k^{3}/\dPsi$,
	$k/\dPsi\geq10^{9}$ for some $\dPsi\in(0,1/2]$ and
	the optimum $k$-means cost of the normalized SE $\Yp$ is such that~\footnote{
		$\left\Vert \Yp -\XpOpt(\XpOpt)^{\rot}\Yp \right\Vert_{F}\geq n^{-O(1)}$ 
		asserts a multiplicative approximation guarantee in \thmref{thmMyPWM}.}
	$\left\Vert \Yp -\XpOpt(\XpOpt)^{\rot}\Yp \right\Vert_{F}\geq n^{-O(1)}$. 
	Compute the approximate normalized SE $\wYp$ via the Power method 
	with $p\geq\Omega(\frac{\ln n}{\lkp})$.
	Then, for $\eps=6\cdot10^{-7}$ it holds with high probability that 
	$\twXVs{k}<5\eps^{2}\cdot\twXVs{k-1}$.
}
\smallskip\smallskip

In Subsection~\ref{subsec:ThmPartThree}, we show that Part (b) of \thmref{myPSZ15}
follows by combining Part (a) of Theorem~\ref{thm:myPSZ15}, Theorem~\ref{thm_KMs}, 
\thmref{thmMyPWM} and \thmref{thmEasySpectralEmbedding}.

\subsubsection{Proof of Lemma~\ref{my_lem_PWM_FN}}\label{subsec:lemMyPWM}

We argue in a similar manner as in \cite[Lemma 7]{BM14}.
Our analysis uses the following two probabilistic results on Gaussian matrices.

\begin{lem}[Norm of a Gaussian Matrix~\cite{DS01}]\label{lem:NormGM}
	Let $M\in\R^{n\times k}$ be a matrix of i.i.d. standard Gaussian random variables, where $n\geq k$. Then, for $t\geq 4$, $\Pr\{\sigma_1(M)\geq t\sqrt{n}\}\geq\exp\{-nt^2/8\}$.
\end{lem}

\begin{lem}[Invertibility of a Gaussian Matrix~\cite{SankarST06}]\label{lem:InvertGM}
	Let $M\in\R^{n\times n}$ be a matrix of i.i.d. standard Gaussian random variables. Then, for any $\beta\in(0,1)$, $\Pr\{\sigma_n(M)\leq\beta/(2.35\sqrt{n})\}\leq\beta$.
\end{lem}

Using the preceding two lemmas, we obtain the following probabilistic statement.

\begin{lem}[Rectangular Gaussian Matrix]\label{lem:recGM}
	Let $S\in\R^{n\times k}$ be a matrix of i.i.d. standard Gaussian random variables, $V\in\R^{n\times\rho}$ be a matrix with orthonormal columns and $n\geq\rho\geq k$. Then, with probability at least $1-e^{-2n}$ it holds $\mathrm{rank}(V^{\rot}S)=k$.
\end{lem}

\begin{proof}
	Let $S^\prm\in\R^{n\times\rho}$ be an extension of $S$ such that $S^\prm=[S\,\,S^{\prm\prm}]$, where $S^{\prm\prm}\in\R^{n\times\rho-k}$ is a matrix of i.i.d. standard Gaussian random variables. Notice that $V^{\rot}S^\prm\in\R^{\rho\times\rho}$ is a matrix of i.i.d. standard Gaussian random variables. We apply now Lemma~\ref{lem:InvertGM} with $\beta=e^{-2n}$ which yields with probability at least $1-e^{-2n}$ that $\sigma_\rho(V^{\rot}S^\prm)>1/(2.35\cdot e^{2n}\sqrt{\rho})>0$ and thus $\mathrm{rank}(V^{\rot}S^\prm)=\rho$. In particular, $\mathrm{rank}(V^{\rot}S)=k$ with probability at least $1-e^{-2n}$.
\end{proof}

\begin{proof}[Proof of Lemma~\ref{my_lem_PWM_FN}]
By the Eigendecomposition theorem, $\LG=U\Sigma'U^{-1}$ where $U\in\R^{n\times n}$
is an orthonormal matrix whose $i$-th column equals the eigenvector of $\LG$ 
corresponding to the $i$-th smallest eigenvalue $\lambda_{i}$, and $\Sigma'$ is a 
non-negative diagonal matrix such that $\Sigma'_{ii}=\lambda_{i}$, for all $i$.
Since the canonical SE $Y\in\R^{n\times k}$ consists of the bottom $k$ 
eigenvectors of $\LG$, we have $U=[Y\,\,\Unk]$ where $\Unk\in\R^{n\times n-k}$,
and similarly $\Sigma=[\Sigma_{k}\,\, 0_{k,n-k};\,\, 0_{n-k,k}\,\,\Sigma_{n-k} ]$.

Further, by the Eigendecomposition theorem $M=U\Sigma U^{\rot}$, where
$\Sigma=2I-\Sigma'$ and in particular $\Sigma_{ii}=2-\lambda_{i}\geq0$ for all $i$.
Since $M^{p}=U\Sigma^{p}U^{\rot}$, it follows that $\mathrm{ker}(M^pS)=\mathrm{ker}(U^{\rot}S)$. 
By Lemma~\ref{lem:recGM} with probability at least $1-e^{-2n}$ we have $\mathrm{rank}(U^{\rot}S)=k$ and thus matrix $M^pS$ has $k$ singular values. Further, the SVD decomposition $\wU \widetilde{\Sigma}\widetilde{V}^{\rot}$ of $M^pS$ satisfies: $\wU \in\R^{n\times k}$ is a matrix with orthonormal columns, $\widetilde{\Sigma}\in\R^{k\times k}$ is a positive diagonal matrix and $\widetilde{V}^{\rot}\in\R^{k\times k}$ is an orthonormal matrix. 
Recall that approximate SE is defined by $\wY=\wU$.

Let $R\Def\widetilde{\Sigma}\widetilde{V}^{\rot}\in\mathbb{R}^{k\times k}$ and observe that
$\wY R=M^pS=[Y\,\, \Unk]\Sigma^p[Y^{\rot};\,\, \Unk^{\rot}]S$.
We use the facts:
\begin{eqnarray}
\wY R & = & Y\Sigma_{k}^{p}Y^{\rot}S+\Unk\Sigma_{n-k}^{p}\Unk^{\rot}S; \label{eq:QR}\\
\sigma_{i}(\wY R) & \geq & \sigma_{k}\left(Y\Sigma_{k}^{p}Y^{\rot}S\right)\geq\left(2-\lk \right)^{p}\cdot\sigma_{k}\left(Y^{\rot}S\right); \label{eq:LBsigmaQR}\\
\sigma_{i}(\wY R) & = & \sigma_{i}\left(R\right); \label{eq:sigmaQR}\\
\lVert X\wY \rVert_{2} & \geq & \lVert X\wY \rVert_{2}\cdot\sigma_{k}(\wY),
\quad\text{for any }X\in\mathbb{R}^{\ell\times k}.\label{eq:XR}
\end{eqnarray}
(\ref{eq:QR}) follows from the eigenvalue decomposition of $M$
and the fact that $M^{p}=U\Sigma^{p}U^{\rot}$; (\ref{eq:LBsigmaQR})
follows by (\ref{eq:QR}) due to $Y$ and $\Unk$ span orthogonal
spaces, and since the minimum singular value of a product is at least
the product of the minimum singular values; (\ref{eq:sigmaQR}) holds
due to $\wY^{\rot}\wY =I_{k}$; Recall that with probability at least $1-e^{-2n}$ we have $\sigma_{k}(R)>0$ and hence (\ref{eq:XR}) follows by
\[
\left\Vert X\right\Vert_{2}=\max_{x\neq0}\frac{\left\Vert XRx\right\Vert_{2}}{\left\Vert Rx\right\Vert_{2}}\leq\max_{x\neq0}\frac{\left\Vert XRx\right\Vert_{2}}{\sigma_{k}(R)\left\Vert x\right\Vert_{2}}=\frac{\left\Vert XR\right\Vert_{2}}{\sigma_{k}(R)}.
\]
\cite[Theorem 2.6.1]{GL12} shows that for every two $m\times k$
orthonormal matrices $W,Z$ with $m\geq k$ it holds
\[
\left\Vert WW^{\rot}-ZZ^{\rot}\right\Vert_{2}=\left\Vert Z^{\rot}W^{\perp}\right\Vert_{2}=\left\Vert W^{\rot}Z^{\perp}\right\Vert_{2},
\]
where $[Z,Z^{\perp}]\in\mathbb{R}^{m\times m}$ is full orthonormal
basis. Therefore, we have
\begin{equation}
\left\Vert YY^{\rot}-\wY \wY^{\rot}\right\Vert_{2}=\left\Vert \wY^{\rot}Y^{\perp}\right\Vert_{2}=\left\Vert (Y^{\perp})^{\rot}\wY \right\Vert_{2}=
\left\Vert \Unk^{\rot}\wY \right\Vert_{2},\label{eq:UkUkT}
\end{equation}
where the last equality is due to $Y^{\perp}=U_{n-k}$.

To upper bound $\left\Vert \Unk^{\rot}\wY \right\Vert_{2}$
we establish the following inequalities:
\begin{eqnarray}
\left\Vert \Unk^{\rot}\wY R\right\Vert_{2} & \geq & \left\Vert \Unk^{\rot}\wY \right\Vert_{2}\cdot\sigma(R)\geq\left\Vert \Unk^{\rot}\wY \right\Vert_{2}\cdot\left(2-\lk \right)^{p}\cdot\sigma_{k}\left(Y^{\rot}S\right),\label{eq:lbUQR}\\
\left\Vert \Unk^{\rot}\wY R\right\Vert_{2} & = & \left\Vert \Sigma_{n-k}^{p}\Unk^{\rot}S\right\Vert_{2}\leq\left(2-\lambda_{k+1}\right)^{p}\cdot\sigma_{1}\left(\Unk^{\rot}S\right),\label{eq:ubUQR}
\end{eqnarray}
where (\ref{eq:lbUQR}) follows by (\ref{eq:XR}), (\ref{eq:sigmaQR})
and (\ref{eq:LBsigmaQR}); and (\ref{eq:ubUQR})  is due to (\ref{eq:QR})
and $2=\Sigma_{11}\geq\cdots\geq\Sigma_{nn}\geq0$.

By Lemma~\ref{lem:NormGM} and Lemma~\ref{lem:InvertGM}, it follows by
the Union bound that with probability at least $1-e^{-2n}-3\beta$, we have
\begin{equation}
\frac{\beta}{\sqrt{k}}\leq\sigma_{k}\left(Y^{\rot}S\right)
\quad\text{and}\quad\sigma_{1}\left(\Unk^{\rot}S\right)\leq4\sqrt{n}.\label{eq:LUB_sigmas}
\end{equation}
Using (\ref{eq:UkUkT}), (\ref{eq:lbUQR}), (\ref{eq:ubUQR})
and (\ref{eq:LUB_sigmas}) we obtain
\begin{equation}
\left\Vert YY^{\rot}-\wY \wY^{\rot}\right\Vert_{2}=\left\Vert \Unk^{\rot}\wY \right\Vert_{2}\leq (4/\beta)\cdot\sqrt{nk}\cdot\gamma_{k}^{p}.\label{eq:UkUkT_L2}
\end{equation}
Since $\left\Vert M\right\Vert_{F}\leq\sqrt{\text{rank}(M)}\cdot\left\Vert M\right\Vert_{2}$
for every matrix $M$ and $\text{rank}(YY^{\rot}-\wY \wY^{\rot})\leq2k$,
it follows
\[
\left\Vert YY^{\rot}-\wY \wY^{\rot}\right\Vert_{F}
\leq 2k\cdot\left\Vert YY^{\rot}-\wY \wY^{\rot}\right\Vert_{2}
\leq (8/\beta)\cdot n^{1/2}k^{3/2} \cdot\gamma_{k}^{p}\leq\epsilon,
\]
where the last two inequalities are due to (\ref{eq:UkUkT_L2}) and
the choice of $\gamma_{k}$.
\end{proof}

\subsubsection{Proof of \thmref{thmMyPWM}}\label{subsec:thmMyPWM}

\cite[Theorem 6]{BoutsidisKG15} relates canonical SE and approximate SE,
whereas our goal is to establish similar result for the normalized SE and 
approximate normalized SE.
We present next four technical lemmas that combined with Lemma~\ref{my_lem_PWM_FN},
allow us to apply the proof technique developed in~\cite[Theorem 6]{BoutsidisKG15}.

\begin{lem}\label{lem_XpXpTisProj}
	Let $\Xp,\wXp\in\R^{m\times k}$ be indicator matrices returned
	by an $\alpha$-approximate $k$-means clustering algorithm 
	applied on inputs $\Yp$ and $\wYp$, respectively, for any $\alpha\geq1$.
	Then, it holds that $\Xp\XpT$ and $\wXp\wXpT$ are projection matrices.
\end{lem}

\begin{proof}
	We prove now the first conclusion.
	By construction, there are $d(v)$ many copies of row $Y(v,:)/\sqrt{d(v)}$ in $\Yp$, 
	for all $v\in V$. 
	W.l.o.g. the indicator matrix $\Xp$ has all copies of row $Y(v,:)/\sqrt{d(v)}$
	assigned to the same cluster, for all $v\in V$. 
	By definition, $\Xp_{ij}=1/\sqrt{\mu(C_j)}$ if row $\Yp_{i,:}$ belongs to 
	the $j$-th cluster $C_j$ and $\Xp_{ij}=0$ otherwise. 
	Hence, $\XpT\Xp=I_{k\times k}$ and thus $[\Xp\XpT]^2=\Xp\XpT$.
	The second part follows similarly, since matrix $\wU$ is orthonormal.
\end{proof}

\begin{lem}\label{lem_Ikk}
	The normalized SE $\Yp$ and 
	the approximate normalized SE $\wYp$
	are orthonormal matrices.
\end{lem}

\begin{proof}
	We prove now $\YpT\Yp =I_{k\times k}$. 
	The equality $\wYp^{\rot}\wYp=I_{k\times k}$ follows similarly. 
	Note that
	\begin{eqnarray*}
		\left[\YpT\Yp \right]_{ij} & = & \left(\begin{array}{ccc}
			\frac{Y(1,i)}{\sqrt{d(1)}}\vec{1}_{d(1)}^{\rot} & \cdots & \frac{Y(n,i)}{\sqrt{d(n)}}\vec{1}_{d(n)}^{\rot}\end{array}\right)\left(\begin{array}{c}
			\frac{Y(1,j)}{\sqrt{d(1)}}\vec{1}_{d(1)}\\
			\cdots\\
			\frac{Y(n,j)}{\sqrt{d(n)}}\vec{1}_{d(n)}
		\end{array}\right)\\
		& = & \sum_{\ell=1}^{n}d(\ell)\frac{Y(\ell,i)}{\sqrt{d(\ell)}}\frac{Y(\ell,j)}{\sqrt{d(\ell)}}=\left\langle Y(:,i),Y(:,j)\right\rangle =\delta(i,j),
	\end{eqnarray*}
	where $\delta(i,j)$ is the Kronecker delta function. Hence, the statement follows.
\end{proof}

\begin{lem}\label{lem_YY_YpYp} 
	It holds that $\lVert \Yp \YpT-\wYp(\wYp)^{\rot}\rVert_{F}=
	\lVert YY^{\rot}-\wY\wY^{\rot}\rVert_{F}$.
\end{lem}

\begin{proof}
	Let $\vec{1}_{d(i)}\in\{0,1\}^{m}$ be an indicator vector of
	the $d(i)$ copies of row $Y(i,:)/\sqrt{d(i)}$ in matrix $\Yp$.
	By definition
	\[
		\Yp \YpT=\sum_{\ell=1}^{k}Y_{:,\ell}^{\prm}Y_{:,\ell}^{\prm T}\quad\text{where}\quad Y_{:,\ell}^{\prm}=\left(\begin{array}{c}
		\frac{Y(1,\ell)}{\sqrt{d(1)}}\vec{1}_{d(1)}\\
		\cdots\\
		\frac{Y(n,\ell)}{\sqrt{d(n)}}\vec{1}_{d(n)}
		\end{array}\right)_{m\times1}
	\]
	and
	\[
		\left(Y_{:,\ell}^{\prm}Y_{:,\ell}^{\prm T}\right)_{d(i)d(j)}=
		\frac{Y(i,\ell)Y(j,\ell)}{\sqrt{d(i)d(j)}}\cdot\vec{1}_{d(i)}\vec{1}_{d(j)}^{\rot}.
	\]
	Hence, we have
	\begin{eqnarray*}
		\left\Vert \Yp\YpT-\wYp(\wYp)^{\rot}\right\Vert _{F}^{2}&=&\sum_{i=1}^{n}\sum_{j=1}^{n}\left\Vert \left(\Yp\YpT-\wYp(\wYp)^{\rot}\right)_{d(i)d(j)}\right\Vert _{F}^{2}\\&=&\sum_{i=1}^{n}\sum_{j=1}^{n}\left\Vert \sum_{\ell=1}^{k}\left(Y_{:,\ell}^{\prm}(Y_{:,\ell}^{\prm})^{\rot}-\wYp_{:,\ell}(\wYp_{:,\ell})^{\rot}\right)_{d(i)d(j)}\right\Vert _{F}^{2}\\&=&\sum_{i=1}^{n}\sum_{j=1}^{n}\left\Vert \left\{ \sum_{\ell=1}^{k}\left(\frac{Y(i,\ell)Y(j,\ell)}{\sqrt{d(i)d(j)}}-\frac{\wY(i,\ell)\wY(j,\ell)}{\sqrt{d(i)d(j)}}\right)\right\} \cdot\vec{1}_{d(i)}\vec{1}_{d(j)}^{\rot}\right\Vert _{F}^{2}.
	\end{eqnarray*}
	By definition of Frobenius norm, (see Subsection~\ref{subsec:Notation}), it holds
	\begin{eqnarray*}		
		\left\Vert \Yp \YpT-\wYp(\wYp)^{\rot}\right\Vert_{F}^{2} & = & \sum_{i=1}^{n}\sum_{j=1}^{n}d(i)d(j)\left[\sum_{\ell=1}^{k}\left(\frac{Y(i,\ell)Y(j,\ell)}{\sqrt{d(i)d(j)}}-\frac{\wY (i,\ell)\wY (j,\ell)}{\sqrt{d(i)d(j)}}\right)\right]^{2}\\
		& = & \sum_{i=1}^{n}\sum_{j=1}^{n}\left[\sum_{\ell=1}^{k}\left(Y(i,\ell)Y(j,\ell)-\wY (i,\ell)\wY (j,\ell)\right)\right]^{2}\\
		& = & \sum_{i=1}^{n}\sum_{j=1}^{n}\left(YY^{\rot}-\wY \wY ^{\rot}\right)_{ij}^{2}\\
		& = & \left\Vert YY^{\rot}-\wY \wY ^{\rot}\right\Vert_{F}^{2}.
	\end{eqnarray*}
\end{proof}

\begin{lem}\label{lem_AATUUT}
	For any matrix $U$ with orthonormal columns and
	every matrix $A$ it holds
	\begin{equation}
	\left\Vert UU^{\rot}-AA^{\rot}UU^{\rot}\right\Vert_{F}=\left\Vert U-AA^{\rot}U\right\Vert_{F}.\label{eq:UUT}
	\end{equation}
\end{lem}

\begin{proof}
	The statement follows by the Frobenius norm property $\left\Vert M \right\Vert_{F}^2=\Tr[M^{\rot}M]$, the cyclic property of trace $\Tr[UM^{\rot}MU^{\rot}]=\Tr[M^{\rot}M\cdot U^{\rot}U]$ and the orthogonality of matrix $U$.
\end{proof}

Using the preceding lemmas, we are ready to prove \thmref{thmMyPWM}.

\begin{proof}[Proof of \thmref{thmMyPWM}]
Using Lemma \ref{my_lem_PWM_FN} and Lemma \ref{lem_YY_YpYp} with probability
at least $1-2e^{-2n}-3\delta_{p}$ we have
\[
\left\Vert \Yp \YpT-\wYp(\wYp)^{\rot}\right\Vert_{F}=\left\Vert YY^{\rot}-\wY\wY^{\rot}\right\Vert_{F}\leq\eps.
\]
Let $\Yp \YpT=\wYp(\wYp)^{\rot}+E$
such that $\left\Vert E\right\Vert_{F}\leq\eps$. By combining Lemma
\ref{lem_Ikk} and Lemma~\ref{lem_AATUUT}, (\ref{eq:UUT}) holds for the matrices $\Yp $ and $\wYp$. Thus, by Lemma~\ref{lem_XpXpTisProj} and the proof techniques in \cite[Theorem 6]{BoutsidisKG15},
it follows that
\begin{equation}\label{eq:myApprKMs}
\left\Vert \Yp -\wXaP(\wXaP)^{\rot}\Yp \right\Vert_{F}\leq\sqrt{\alpha}\cdot\left(\left\Vert \Yp -\XpOpt(\XpOpt)^{\rot}\Yp \right\Vert_{F}+2\eps\right).
\end{equation}
The desired statement follows by simple algebraic manipulations of (\ref{eq:myApprKMs}).
\end{proof}

\subsubsection{Proof of \thmref{thmEasySpectralEmbedding}}\label{subsec:ThmPartTwo}

In this Subsection, we demonstrate that the 
approximate normalized SE $\wYp$
is $\eps$-separated, i.e. $\twXVs{k}<5\eps^{2}\cdot\twXVs{k-1}$.
Our analysis builds upon \thmref{thmGapTriK}, \thmref{thmMyPWM} and the proof techniques 
in~\cite[Theorem 6]{BoutsidisKG15}.

Before we present the proof of \thmref{thmEasySpectralEmbedding}, we establish two technical Lemmas.

\begin{lem}\label{lem_lnLB}
	Suppose $\Psi\geq20^{4}\cdot k^{3}/\dPsi$ for some $\dPsi\in(0,1/2]$.
	Then, it holds that
	\[
		\ln\left(\frac{2-\lk }{2-\lkp}\right)\geq \frac{1}{2}\left(1-\frac{4\dPsi}{20^{4}k^{2}}\right)\lkp.
	\]
\end{lem}

\begin{proof}
	By \eqref{eq:highorder}, the following higher-order Cheeger inequalities hold
	\begin{equation}
	\lk /2\leq\rho(k)\leq O(k^{2})\cdot\sqrt{\lk }.\label{eq:HOCIeq}
	\end{equation}
	Using the LHS of (\ref{eq:HOCIeq}), we have
	\[
	k^{3}\hrAvrK=k^{2}\sum_{i=1}^{k}\phi(P_{i})\geq k^{2}\max_{i\in\{1,\dots,k\}}\phi(P_{i})\geq k^{2}\cdot\rho(k)\geq\frac{k^{2}\lk }{2},
	\]
	and thus the $k$-th smallest eigenvalue of $\LG$ satisfies $\lk \leq2k\cdot\hrAvrK$. 
	Then, the gap assumption yields
	\[
		\lkp\geq\frac{20^{4}k^{2}}{2\dPsi}\cdot2k\cdot\hrAvrK\geq\frac{20^{4}k^{2}}{2\dPsi}\cdot\lk .
	\]
	The statement follows by
	\[
		\frac{2-\lk }{2-\lambda_{k+1}}\geq\frac{1-\frac{\delta}{20^{4}k^{2}}\cdot
		\lambda_{k+1}}{1-\frac{1}{2}\lambda_{k+1}}\geq\exp\left\{ \frac{1}{2}\left(1-\frac{4\delta}{20^{4}k^{2}}\right)\lambda_{k+1}\right\}.\hfill\qedhere
	\]
\end{proof}

\begin{lem}\label{lem:MtxNormIneq}
	For any matrices $A\in\R^{m\times n}$ and $B\in\R^{n\times k}$, it holds
	that $\| AB \|_F \leq \|A\|_2\cdot\|B\|_F$.
\end{lem}
\begin{proof}
	By definition, $\|B\|_F^2=\sum_{i=1}^{k}\|B_{:,i}\|_2^2$ and $\|Ax\|_2\leq\|A\|_2\|x\|_2$,
	and thus we have
	\[
		\| AB \|_F^2 = \sum_{i=1}^{k}\|AB_{:,i}\|_2^2 
		\leq \|A\|_2^2\sum_{i=1}^{k}\|B_{:,i}\|_2^2
		= \|A\|_2^2\cdot\|B\|_2^2.\hfill\qedhere
	\]
\end{proof}

In the following, we use interchangeably $\XpOpt$ and $\XOptPrm{k}$ 
to denote an optimal indicator matrix for the $k$-means clustering problem on $\Yp$.
Similarly, we denote by $\XOptPrm{k-1}$ an optimal indicator matrix for the 
$(k-1)$-means clustering problem on $\Yp$.

\begin{cor}\label{cor_GapKMUB}
	Let $G$ be a graph that satisfies $\Psi=20^{4}\cdot k^{3}/\dPsi$,
	$\dPsi\in(0,1/2]$ and $k/\dPsi\geq10^{9}$. Then, it holds that
	\[
		\left\Vert \Yp-\XOptPrm{k}\Big(\XOptPrm{k}\Big)^{\rot}\Yp\right\Vert _{F}^{2}\leq\frac{1}{8\cdot10^{13}}.
	\]
\end{cor}
\begin{proof}
	The statement follows by Lemma~\ref{upper bound on the cost of S}.
\end{proof}

We are now ready to prove \thmref{thmEasySpectralEmbedding}.

\begin{proof}[Proof of \thmref{thmEasySpectralEmbedding}]
By \thmref{thmGapTriK}, we have
\begin{equation}\label{eq:Thm5Bound}
	\left\Vert \Yp-\XOptPrm{k}\Big(\XOptPrm{k}\Big)^{\rot}\Yp\right\Vert _{F}\leq\eps\left\Vert \Yp-\XOptPrm{k-1}\Big(\XOptPrm{k-1}\Big)^{\rot}\Yp\right\Vert _{F}.
\end{equation}
We now set the approximation parameter $\eps^{\prm}$ in \thmref{thmMyPWM},
and using \eqref{eq:Thm5Bound} we upper bound it by
\begin{equation}\label{eq:epsPrmDef}
	\eps^{\prm}\Def n^{-O(1)}\leq\frac{1}{4}\left\Vert \Yp-\XOptPrm{k}\Big(\XOptPrm{k}\Big)^{\rot}\Yp\right\Vert _{F}=\frac{1}{4}\sqrt{\tXV{k}}\leq\frac{\eps}{4}\sqrt{\tXV{k-1}}.
\end{equation}
The approximate SE $\wY\in\R^{n\times k}$, see \eqref{eq:defWY},
is constructed via the Power method applied with
$p\geq\Omega(\frac{\ln n}{\lkp})$. 
By combining Lemma~\ref{my_lem_PWM_FN} and Lemma~\ref{lem_YY_YpYp},
for the normalized and approximate normalized SE,
$\Yp$ and $\wYp$ respectively, we obtain w.h.p. that 
\[
\left\Vert \Yp \Yp {}^{\rot}-\wYp(\wYp)^{\rot}\right\Vert_{F}=\left\Vert YY^{\rot}-\wY\wY^{\rot}\right\Vert_{F}\leq\eps^{\prm}.
\]
Let $\Yp \Yp {}^{\rot}=\wYp\wYp^{\rot}+E$
such that $\left\Vert E\right\Vert_{F}\leq\eps^{\prm}$. 
By Lemma~\ref{lem_Ikk}, $\Yp$ and $\wYp$ are orthonormal matrices. 
Hence, by Lemma~\ref{lem_AATUUT} applied on $\wYp$, we obtain
\begin{eqnarray}
	\sqrt{\twXVs{k}}&=&
	\left\Vert \wYp-\widetilde{\XOptPrm{k}}\Big(\widetilde{\XOptPrm{k}}\Big)^{\rot}\wYp\right\Vert _{F}=\left\Vert \wYp\wYp^{\rot}-\widetilde{\XOptPrm{k}}\Big(\widetilde{\XOptPrm{k}}\Big)^{\rot}\wYp\wYp^{\rot}\right\Vert _{F}\nonumber\\
	&=&\left\Vert \Yp\YpT-\widetilde{\XOptPrm{k}}\Big(\widetilde{\XOptPrm{k}}\Big)^{\rot}\Yp\YpT-\Big[I_{m\times m}-\widetilde{\XOptPrm{k}}\Big(\widetilde{\XOptPrm{k}}\Big)^{\rot}\Big]E\right\Vert _{F}\nonumber\\
	&\leq&\left\Vert \Yp-\widetilde{\XOptPrm{k}}\Big(\widetilde{\XOptPrm{k}}\Big)^{\rot}\Yp\right\Vert _{F}+\left\Vert E\right\Vert _{F},\label{eq:tkXvEF}
\end{eqnarray}
where the last step uses triangle inequality, Lemma~\ref{lem_AATUUT} applied on $\Yp$, Lemma~\ref{lem:MtxNormIneq}, Lemma~\ref{lem_XpXpTisProj} and $\|I-PP^{\rot}\|_2\leq1$ 
for any projection matrix $P$.
Then, we apply \thmref{thmMyPWM} with an exact $k$-means clustering algorithm, i.e. $\alpha=1$, 
$\delta_p=n^{-O(1)}$, $\eps^{\prm} = n^{-O(1)}$ and by Lemma~\ref{lem_lnLB} 
for $p\geq\Omega(\frac{\ln n}{\lkp})$ as above, we obtain with high probability
\begin{equation}\label{eq:optApproxSolBnd}
	\left\Vert \Yp-\widetilde{\XOptPrm{k}}\Big(\widetilde{\XOptPrm{k}}\Big)^{\rot}\Yp\right\Vert _{F}^{2}\leq(1+4\eps^{\prm})\cdot\left\Vert \Yp-\XOptPrm{k}\Big(\XOptPrm{k}\Big)^{\rot}\Yp\right\Vert _{F}^{2}+4(\eps^{\prm})^{2}.
\end{equation}
Then, by combining \eqref{eq:tkXvEF}, \eqref{eq:optApproxSolBnd},
$\left\Vert E\right\Vert_{F}\leq\eps^{\prm}$ and the LHS of \eqref{eq:epsPrmDef},
we have
\begin{eqnarray}
\sqrt{\twXVs{k}}
&\leq&\eps^{\prm}+\sqrt{(1+4\eps^{\prm})\left\Vert \Yp-\XOptPrm{k}\Big(\XOptPrm{k}\Big)^{\rot}\Yp\right\Vert _{F}^{2}+4(\eps^{\prm})^{2}}\nonumber\\
&\leq&2\left\Vert \Yp-\XOptPrm{k}\Big(\XOptPrm{k}\Big)^{\rot}\Yp\right\Vert _{F}\nonumber\\
&\leq&2\eps\cdot\sqrt{\tXV{k-1}},\label{eq:TriOne}
\end{eqnarray}
where the last inequality follows by \eqref{eq:Thm5Bound}. Moreover, it holds that
\begin{eqnarray*}
	\sqrt{\tXV{k-1}}&=&\left\Vert \Yp-\XOptPrm{k-1}\Big(\XOptPrm{k-1}\Big)^{\rot}\Yp\right\Vert _{F}\leq\left\Vert \Yp-\widetilde{\XOptPrm{k-1}}\Big(\widetilde{\XOptPrm{k-1}}\Big)^{\rot}\Yp\right\Vert _{F}\\&=&\left\Vert \Yp\YpT-\widetilde{\XOptPrm{k-1}}\Big(\widetilde{\XOptPrm{k-1}}\Big)^{\rot}\Yp\YpT\right\Vert _{F}\\&=&\left\Vert \wYp\wYp^{\rot}-\widetilde{\XOptPrm{k-1}}\Big(\widetilde{\XOptPrm{k-1}}\Big)^{\rot}\wYp\wYp^{\rot}+\Big[I_{m\times m}-\widetilde{\XOptPrm{k-1}}\Big(\widetilde{\XOptPrm{k-1}}\Big)^{\rot}\Big]E\right\Vert _{F}\\&\leq&\left\Vert \wYp-\widetilde{\XOptPrm{k-1}}\Big(\widetilde{\XOptPrm{k-1}}\Big)^{\rot}\wYp\right\Vert _{F}+\left\Vert E\right\Vert _{F}\\&\leq&\sqrt{\twXVs{k-1}}+\frac{\eps}{4}\sqrt{\tXV{k-1}},
\end{eqnarray*}
where the last inequality uses $\left\Vert E\right\Vert_{F}\leq\eps^{\prm}$
and \eqref{eq:epsPrmDef}. Hence,
\begin{equation}
\sqrt{\tXV{k-1}}\leq\left(1+\frac{\eps}{2}\right)\sqrt{\twXVs{k-1}}.\label{eq:TriTwo}
\end{equation}
The statement follows by combining (\ref{eq:TriOne}) and ($\ref{eq:TriTwo}$), i.e.
\[
	\sqrt{\twXVs{k}}\leq2\eps\cdot\sqrt{\tXV{k-1}}\leq
	(2+\eps)\cdot\eps\cdot\sqrt{\twXVs{k-1}}.\hfill\qedhere
\]
\end{proof}

\subsection{Proof of Approximate Spectral Clustering}\label{subsec:ThmPartThree}

We prove now Part (b) of \thmref{myPSZ15}.
Let $p=\Theta(\frac{\ln n}{\lkp})$. We compute the matrix $M^{p}S$ in time $O(mkp)$ and 
its singular value decomposition $\wU \widetilde{\Sigma}\widetilde{V}^{\rot}$ in time
$O(nk^{2})$. Based on it, we construct in time $O(mk)$ 
the approximate normalized SE $\wYp$, see \eqref{eq:defYpWYp}.

By \thmref{thmEasySpectralEmbedding}, $\wYp$
is $\eps$-separated for $\eps=6\cdot10^{-7}$, i.e. $\twXVs{k}<5\eps^{2}\cdot\twXVs{k-1}.$
Let $\alpha=1+10^{-10}$. Then, by Theorem \ref{thm_KMs}, there is an algorithm that 
outputs in time $O(mk^{2}+k^{4})$ a $k$-way partition with indicator matrix $\wXaP$ 
such that with probability at least $1-O(\sqrt{\eps})$, we have
\[
\left\Vert \wYp-\wXaP\Big(\wXaP\Big)^{\rot}\wYp\right\Vert _{F}^{2}\leq\left(1+\frac{1}{10^{10}}\right)\cdot\left\Vert \wYp-\wXpOpt\Big(\wXpOpt\Big)^{\rot}\wYp\right\Vert _{F}^{2}.
\]

Let $\eta\in(n^{-O(1)},1)$ be a parameter to be determined soon.
By \thmref{thmEasySpectralEmbedding} and Corollary~\ref{cor_GapKMUB}, we have
\[
\frac{1}{n^{O(1)}}\leq\left\Vert \Yp -\XpOpt\Big(\XpOpt\Big)^{\rot}\Yp \right\Vert_{F}\leq\frac{1}{10^{6}}.
\]
Using Lemma~\ref{lem_lnLB}, we apply \thmref{thmMyPWM} with $\delta_p=n^{-O(1)}$,
$\alpha=1+10^{-10}$, $\delta_{\alpha}=O(\sqrt{\eps})$ and 
\[
	\eps^{\prm} = \frac{\sqrt{\eta}}{4}\cdot \frac{1}{n^{O(1)}}
	\leq \frac{\sqrt{\eta}}{4}\cdot 
	\left\Vert \Yp-\XpOpt\Big(\XpOpt\Big)^{\rot}\Yp\right\Vert _{F},
\]
and obtain with constant probability (close to 1) that
\begin{eqnarray*}
	&&\left\Vert \wYp-\wXaP\Big(\wXaP\Big)^{\rot}\wYp\right\Vert _{F}^{2}\leq(1+4\eps^{\prm})\cdot\alpha\cdot\left\Vert \Yp-\XpOpt\Big(\XpOpt\Big)^{\rot}\Yp\right\Vert _{F}^{2}+4\eps^{\prm2}\\&=&\left[\left(1+\sqrt{\eta}\left\Vert \Yp-\XpOpt\Big(\XpOpt\Big)^{\rot}\Yp\right\Vert _{F}\right)\alpha+\frac{\eta}{4}\right]\cdot\left\Vert \Yp-\XpOpt\Big(\XpOpt\Big)^{\rot}\Yp\right\Vert _{F}^{2}\\&\leq&\left[\left(1+\frac{\sqrt{\eta}}{10^{6}}\right)\cdot\left(1+\frac{1}{10^{10}}\right)+\frac{\eta}{4}\right]\cdot\left\Vert \Yp-\XpOpt\Big(\XpOpt\Big)^{\rot}\Yp\right\Vert _{F}^{2},
\end{eqnarray*}

Then, for $\eta=1/10^{6}$ the approximate solution $\wXaP$ yields a multiplicative approximation,
satisfying
\[
\left\Vert \Yp-\wXaP\Big(\wXaP\Big)^{\rot}\Yp\right\Vert _{F}^{2}\leq\left(1+\frac{1}{10^{6}}\right)\left\Vert \Yp-\XpOpt\Big(\XpOpt\Big)^{\rot}\Yp\right\Vert _{F}^{2}.
\]
The statement follows by Part (a) of \thmref{myPSZ15} applied to the $k$-way partition
$(A_{1},\dots,A_{k})$ of $V$ that is induced by the indicator matrix $\wXaP$.

	\section{Open Problems}
	Orecchia and Allen Zhu~\cite{OrecchiaZ14} showed that for any node sets $\widehat{S},S$
	if $\widehat{S}$ has a large volume overlap with $S$, i.e.
	$\mu(\widehat{S}\cap S)\geq\delta\mu(S)$ for some $\delta\in(1/2,1)$,
	then there is an efficient ``local flow refinement procedure'' 
	that given $\widehat{S}$ 
	finds a node set $S'$ such that the volume overlap $\mu(S'\cap S)\geq\delta\mu(S)$ 
	and the conductance $\phi(S')\leq O(1/\delta)\phi(S)$.
	
	Let $(A_1,\dots,A_k)$ be the $k$-way node partition computed 
	in Part (b) of Theorem~\ref{thm:myPSZ15}.
	Note that each cluster $A_i$ has a large volume overlap with the 
	corresponding optimal cluster $P_i$.
	However, applying the procedure in~\cite{OrecchiaZ14} to each cluster $A_i$,
	results in a $k$-way node clustering which in general has node-overlapping clusters, 
	and thus breaks the partitioning property.
	
	Hence, an important research direction is to prove or disprove the existence of an
	efficient global refinement procedure, that on input the $k$-way node partition 
	$(A_1,\dots,A_k)$ outputs a refined $k$-way node partition $(B_1,\dots,B_k)$ such that
	$\phi(B_i)\leq O(1)\cdot\phi(P_i)$,	for all $i\in\{1,\dots,k\}$.
	
	Another research direction is to improve the constants in the work of
	Ostrovsky et al.~\cite{ORSS12} and to extend the analysis of
	Theorem~\ref{thm:myPSZ15} to small graphs.

	\section{Acknowledgements}
	The authors want to thank Michael B. Cohen for a helpful discussion 
	that led us to the work of Boutsidis et al.~\cite{BoutsidisKG15}, 
	and anonymous reviewers of our ESA'16 submission for their helpful remarks.

	\bibliographystyle{alpha}
	\bibliography{bibliography}
\end{document}